\documentclass[journal,twoside]{IEEEtran}

\usepackage{mathrsfs}
\usepackage[noadjust]{cite}
\usepackage{graphicx,color,overpic,psfrag}
\usepackage{amsmath, amssymb}
\usepackage{latexsym}
\usepackage{bm}
\usepackage{amssymb}
\usepackage{cases}
\usepackage{array}
\usepackage{fancyhdr}
\usepackage{setspace}
\usepackage{subfigure}
\usepackage{url}
\usepackage{algpseudocode}
\usepackage{algorithm}
\usepackage{blkarray}
\usepackage{booktabs}
\usepackage{multirow}
\usepackage{dsfont}
\usepackage{tabularx}
\usepackage{amsfonts}

\newtheorem{theorem}{Theorem}
\newtheorem{lemma}{Lemma}
\newtheorem{remark}{Remark}

\newtheorem{corol}{Corollary}
\newtheorem{ass}{Assumption}

\newcommand{\asref}[1]{Assumption \ref{#1}}
\newcommand{\lmref}[1]{Lemma \ref{#1}}
\newcommand{\thref}[1]{Theorem \ref{#1}}
\newcommand{\coref}[1]{Corollary \ref{#1}}
\newcommand{\figref}[1]{Fig. \ref{#1}}

\newcommand{\alref}[1]{Algorithm \ref{#1}}
\newcommand{\appref}[1]{Appendix \ref{#1}}
\newcommand{\secref}[1]{Section \ref{#1}}
\newcommand{\remref}[1]{Remark \ref{#1}}

\newcommand{\qedend}{\endIEEEproof}

\newcommand{\Order}[1]{\mathcal{O}(#1)}

\newcommand{\Exp}{{\mathbb{E}}}
\newcommand{\expect}[1]{\Exp\left\{#1\right\}}

\newcommand{\tr}[1]{\mathrm{tr}\left\{#1\right\}}
\newcommand{\diag}[1]{\mathrm{diag}\left\{#1\right\}}
\newcommand{\abs}[1]{\left|#1\right|}

\newcommand{\logtwo}[1]{\log_{2}\left(#1\right)}
\newcommand{\normt}[1]{\left\|#1\right\|_{2}}
\newcommand{\sqrnormt}[1]{\left\|#1\right\|_{2}^{2}}
\newcommand{\normf}[1]{\left\|#1\right\|_{\mathrm{F}}}

\newcommand{\inpro}[2]{\langle #1 , #2\rangle}

\newcommand{\oneon}[1]{\frac{1}{#1}}
\newcommand{\oneldots}[1]{1,2,\ldots,#1} 
\newcommand{\intd}{\mathrm{d}}
\newcommand{\intdx}[1]{\intd #1}

\newcommand{\ppd}[1]{ {\frac{ \partial }{ \partial {#1} }} }
\newcommand{\inv}[1]{#1^{-1}}
\newcommand{\sqrinv}[1]{#1^{-2}}

\newcommand{\GCN}[2]{\mathcal{CN}\left( #1 , #2\right) }
\newcommand{\argmax}[1]{\mathop{\arg\max}\limits_{#1}}
\newcommand{\argmin}[1]{\mathop{\arg\min}\limits_{#1}}
\newcommand{\mini}[1]{\mathop{\mathrm{min}}\limits_{#1}}
\newcommand{\minis}[1]{\mathop{\mathrm{min}}\limits_{#1}}
\newcommand{\st}{\mathrm{subject}\quad\mathrm{to}}
\newcommand{\ith}[1]{\mbox{$#1^{\text{th}}$}}

\newcommand{\toinf}[1]{#1 \to\infty}

\newcommand{\liminfty}[1]{\lim_{\toinf{#1}}}
\newcommand{\delfunc}[1]{\delta\left(#1\right)}
\newcommand{\vecele}[2]{\left[#1\right]_{#2}}

\newcommand{\setdif}[2]{#1\backslash #2}
\newcommand{\eletoset}[1]{\left\{#1\right\}}
\newcommand{\costh}[2]{\cos\theta\left(#1,#2\right)}
\newcommand{\thcos}[2]{\theta\left(#1,#2\right)}

\newcommand{\equaa}{\mathop{=}^{(\textrm{a})}}
\newcommand{\equab}{\mathop{=}^{(\textrm{b})}}
\newcommand{\equac}{\mathop{=}^{(\textrm{c})}}
\newcommand{\equad}{\mathop{=}^{(\textrm{d})}}

\newcommand{\cA}{\mathcal{A}}

\newcommand{\cK}{\mathcal{K}}

\newcommand{\cP}{\mathcal{P}}

\newcommand{\cT}{\mathcal{T}}

\newcommand{\ba}{\mathbf{a}}
\newcommand{\bb}{\mathbf{b}}

\newcommand{\bg}{\mathbf{g}}

\newcommand{\bn}{\mathbf{n}}

\newcommand{\br}{\mathbf{r}}

\newcommand{\bv}{\mathbf{v}}
\newcommand{\bw}{\mathbf{w}}
\newcommand{\bx}{\mathbf{x}}
\newcommand{\by}{\mathbf{y}}

\newcommand{\bA}{\mathbf{A}}
\newcommand{\bB}{\mathbf{B}}
\newcommand{\bC}{\mathbf{C}}

\newcommand{\bG}{\mathbf{G}}

\newcommand{\bI}{\mathbf{I}}

\newcommand{\bN}{\mathbf{N}}

\newcommand{\bR}{\mathbf{R}}

\newcommand{\bV}{\mathbf{V}}
\newcommand{\bW}{\mathbf{W}}
\newcommand{\bX}{\mathbf{X}}
\newcommand{\bY}{\mathbf{Y}}

\newcommand{\invgammaopt}{\frac{1}{\upopt{\gamma}}}

\newcommand{\hatba}{\hat{\ba}}
\newcommand{\hatbg}{\hat{\bg}}

\newcommand{\hatbG}{\hat{\bG}}

\newcommand{\tildebg}{\tilde{\bg}}
\newcommand{\tildebG}{\tilde{\bG}}

\newcommand{\bzero}{\mathbf{0}}

\newcommand{\bOmega}{{\boldsymbol\Omega}}
\newcommand{\bomega}{{\boldsymbol\omega}}

\newcommand{\fortext}{\textrm{for}}

\newcommand{\pktheta}{\mathsf{S}_{k}\left(\theta\right)}

\newcommand{\vtheta}{\bv\left(\theta\right)}

\newcommand{\pklaptheta}{\mathsf{S}_{k}^{\textrm{Lap}}\left(\theta\right)}

\newcommand{\sintheta}{\sin\left(\theta\right)}

\newcommand{\expx}[1]{\exp\left(#1\right)}

\newcommand{\vx}[1]{\bv\left(#1\right)}

\newcommand{\upup}[1]{#1^{\textrm{p}}}
\newcommand{\upt}[1]{#1^{\textrm{t}}}

\newcommand{\upu}[1]{#1^{\textrm{u}}}
\newcommand{\upd}[1]{#1^{\textrm{d}}}

\newcommand{\upopt}[1]{#1^{\textrm{opt}}}

\newcommand{\uptneff}[1]{#1^{\textrm{t,n,eff}}}

\newcommand{\upusum}[1]{#1^{\textrm{u,sum}}}
\newcommand{\updsum}[1]{#1^{\textrm{d,sum}}}
\newcommand{\upnet}[1]{#1^{\textrm{net}}}

\newcommand{\upumin}[1]{#1^{\textrm{u,min}}}
\newcommand{\updmin}[1]{#1^{\textrm{d,min}}}
\newcommand{\uptmin}[1]{#1^{\textrm{t,min}}}
\newcommand{\upnot}[2]{#1^{\textrm{#2}}}
\newcommand{\dnnot}[2]{#1_{\textrm{#2}}}

\newcommand{\thetamin}{\theta^{\textrm{min}}}
\newcommand{\thetamax}{\theta^{\textrm{max}}}

\newcommand{\Pindex}[1]{\pi_{#1}}
\newcommand{\Kset}[1]{\cK_{#1}}
\newcommand{\KPset}[1]{\Kset{\Pindex{#1}}}

\newcommand{\rhoup}{\upup{\rho}}
\newcommand{\rhou}{\upu{\rho}}
\newcommand{\rhod}{\upd{\rho}}
\newcommand{\rhot}{\upt{\rho}}
\newcommand{\invrhod}{\frac{1}{\upd{\rho}}}
\newcommand{\invrhot}{\frac{1}{\upt{\rho}}}
\newcommand{\invsqrtrhod}{\frac{1}{\sqrt{\upd{\rho}}}}

\newcommand{\invrhou}{\frac{1}{\upu{\rho}}}
\newcommand{\invsqrtrhou}{\frac{1}{\sqrt{\upu{\rho}}}}

\newcommand{\invtaurhoup}{\frac{1}{\upup{\rho}\tau}}

\newcommand{\invsqrttaurhoup}{\frac{1}{\sqrt{\upup{\rho}\tau}}}
\newcommand{\invrho}{\frac{1}{\rho}}

\newcommand{\alphaopt}{\alpha^{\textrm{opt}}}

\newcommand{\pktau}{\cP(\cK,\cT)}

\newcommand{\covpi}[1]{\bC_{\Pindex{#1}}}
\newcommand{\invcovpi}[1]{\inv{\covpi{#1}}}

\newcommand{\ypkup}{\upup{\by_{\Pindex{k}}}}
\newcommand{\ypiup}{\upup{\by_{\Pindex{i}}}}
\newcommand{\ypjup}{\upup{\by_{\Pindex{j}}}}
\newcommand{\npkup}{\upup{\bn_{\Pindex{k}}}}

\newcommand{\bnu}{\upu{\bn}}
\newcommand{\bnd}{\upd{\bn}}

\newcommand{\ntb}{\notag\\}

\newcommand{\barjmath}{\bar{\jmath}}

\newcommand{\figdoucolwid}{0.5\textwidth}

\allowdisplaybreaks

\begin{document}

\title{Pilot Reuse for Massive MIMO Transmission over Spatially Correlated Rayleigh Fading Channels}

\author{
Li~You,~\IEEEmembership{Student~Member,~IEEE,} Xiqi~Gao,~\IEEEmembership{Fellow,~IEEE,} Xiang-Gen~Xia,~\IEEEmembership{Fellow,~IEEE,} Ni~Ma,~\IEEEmembership{Senior~Member,~IEEE,} and~Yan~Peng%

\thanks{Copyright \copyright\  2015 IEEE. Personal use of this material is permitted. However, permission to use this material for any other purposes must be obtained from the IEEE by sending a request to pubs-permissions@ieee.org.}%
\thanks{This work was supported by National Natural Science Foundation of China under Grants 61471113, 61320106003, and 61201171, the China High-Tech 863 Plan under Grants 2015AA011305 and 2014AA01A704, National Science and Technology Major Project of China under Grant 2014ZX03003006-003, Huawei Cooperation Project, and the Program for Jiangsu Innovation Team. The work of L. You was supported in part by the China Scholarship Council (CSC). This work was presented in part at the 2014 IEEE International Conference on Communications (ICC) \cite{You14Massive}.
}
\thanks{
L. You and X. Q. Gao are with the National Mobile Communications Research Laboratory, Southeast University, Nanjing 210096, P. R. China (e-mail: liyou@seu.edu.cn; xqgao@seu.edu.cn).
}
\thanks{
X.-G. Xia is with the Department of Electrical and Computer Engineering, University of Delaware, Newark, DE 19716, USA (e-mail: xxia@ee.udel.edu).
}
\thanks{
N. Ma and Y. Peng are with the Huawei Technologies Co., Ltd., Shenzhen 518129, P. R. China (e-mail: ni.ma@huawei.com; yan.peng@huawei.com).
}
}


\markboth{IEEE Transactions on wireless communications}{YOU \MakeLowercase{\textit{et al.}}: PILOT REUSE FOR MASSIVE MIMO TRANSMISSION OVER SPATIALLY CORRELATED RAYLEIGH FADING CHANNELS}

\maketitle

\begin{abstract}
We propose pilot reuse (PR) in single cell for massive multiuser multiple-input multiple-output (MIMO) transmission to reduce the pilot overhead. For spatially correlated Rayleigh fading channels, we establish a relationship between channel spatial correlations and channel power angle spectrum when the base station antenna number tends to infinity. With this channel model, we show that sum mean square error (MSE) of channel estimation can be minimized provided that channel angle of arrival intervals of the user terminals reusing the pilots are non-overlapping, which shows feasibility of PR over spatially correlated massive MIMO channels with constrained channel angular spreads. Regarding that channel estimation performance might degrade due to PR, we also develop the closed-form robust multiuser uplink receiver and downlink precoder that minimize sum MSE of signal detection, and reveal a duality between them. Subsequently, we investigate pilot scheduling, which determines the PR pattern, under two minimum MSE related criteria, and propose a low complexity pilot scheduling algorithm which relies on the channel statistics only. Simulation results show that the proposed PR scheme provides significant performance gains over the conventional orthogonal training scheme in terms of net spectral efficiency.
\end{abstract}


\section{Introduction}

\IEEEPARstart{M}{assive} multiple-input multiple-output (MIMO) transmission employs a large number of antennas at the base station (BS) to serve a relatively smaller number of user terminals (UTs) simultaneously \cite{Marzetta10Noncooperative}. With the potential large gains in spectral efficiency and energy efficiency, massive MIMO is a promising technology that the next generation of wireless systems may incorporate, and has received tremendous research interest recently \cite{Rusek13Scaling,Larsson14Massive}.

Channel state information (CSI) at the BS plays an important role in massive MIMO transmission, and in realistic systems it is typically obtained with assistance of the periodically inserted pilot signals \cite{Tong04Pilot}. In time-division duplex (TDD) massive MIMO transmission, CSI at the BS can be obtained from uplink (UL) training via leveraging the channel reciprocity \cite{Marzetta10Noncooperative,Marzetta06How}. For the conventional orthogonal training (OT) scheme \cite{Marzetta06How}, the pilot overhead is proportional to the number of the UT antennas. As the UT antenna number grows, the heavy pilot overhead decreases the system efficiency greatly and can become the system bottleneck.

In order to reduce the pilot overhead, we propose pilot reuse (PR) in single cell for massive MIMO transmission in this paper. The motivation stems from that, in realistic outdoor wireless propagation environments where BS is located at an elevated position, the scattering around the BS is usually limited, and the MIMO channels are not spatially isotropic \cite{Tse05Fundamentals,Clerckx13MIMO}, i.e., most of the channel power lies in a finite number of spatial directions compared with the whole massive MIMO channel dimension. For UTs with channels lying in almost orthogonal spatial directions, PR is feasible and beneficial.

In the proposed PR scheme, massive MIMO transmission consists of the following phases: statistical CSI acquisition for pilot scheduling, UL training for channel estimation, UL data transmission, and downlink (DL) data transmission. The pilot scheduler at the BS determines the PR pattern, and allocates the available pilot signals to the UTs. Due to the slow-varying nature of the long term channel statistics, it is reasonable to exploit the statistical CSI at the BS to perform pilot scheduling. With the resulting PR pattern, the UTs transmit the respective assigned pilot signals periodically to enable the BS to obtain the channel estimates. The channel estimation performance might degrade due to PR, thus it is natural to design the UL and DL data transmissions robust to the channel estimation error.

In this work, we consider the spatially correlated Rayleigh fading channels, and show that when the BS antenna number tends to infinity, eigenvectors of the channel covariance matrix are determined by the BS array response vectors, while eigenvalues depend on the channel power angle spectrum (PAS), which reveals a relationship between channel spatial correlations and channel power distribution in the angular domain. For this channel model, we show that sum mean square error of channel estimation (MSE-CE) can be minimized, provided that channel angle of arrival (AoA) intervals of the UTs reusing the pilots are non-overlapping, which shows feasibility of PR over spatially correlated massive MIMO channels with constrained channel angular spreads (ASs). Regarding that channel estimation performance might degrade due to PR, we investigate robust data transmissions for both UL and DL with channel estimation error due to PR taken into account. The closed-form robust multiuser UL receiver and DL precoder which are applicable to arbitrary PR pattern based on the minimum MSE of signal detection (MMSE-SD) criterion are developed, and an interesting MMSE duality between them is revealed. Subsequently, we study pilot scheduling under two MMSE related criteria, and in both cases the designs are formulated as combinatorial optimization problems. We show that both criteria can be optimized provided that channel AoA intervals of the UTs reusing the pilots are non-overlapping, and propose a low complexity pilot scheduling algorithm (called the statistical greedy pilot scheduling [SGPS] algorithm) motivated by the channel AoA non-overlapping condition. Simulation results show that the proposed PR scheme provides significant performance gains over the conventional OT scheme in terms of net spectral efficiency.

\emph{Related Works and Our Contributions:} Most of the previous works assumed pilot reuse among cells for massive MIMO transmission, where UTs in the same cell use orthogonal pilots, while the same set of orthogonal pilots is reused among cells \cite{Marzetta10Noncooperative,Hoydis13Massive,Jose11Pilot}. It has been shown that pilot contamination \cite{Jose11Pilot} caused by inter-cell pilot reuse can degrade the performance of massive MIMO transmission. In order to mitigate pilot contamination, several approaches including, e.g., coordinated channel estimation \cite{Yin13coordinated}, time-shifted pilot allocation \cite{Fernandes13Inter}, eigenvalue decomposition based blind channel estimation \cite{Ngo12EVD}, cooperative pilot contamination precoding \cite{Ashikhmin12Pilot}, and distributed MMSE precoding \cite{Jose11Pilot} were proposed, respectively. In contrast to these existing works where pilot overhead was simply set to be fixed, our work focuses on reducing the pilot overhead, and in the meanwhile, balancing the tradeoff between the pilot overhead and the pilot interference.\footnote{Pilot interference caused by intra-cell PR is defined in \eqref{eq:ypkup}, and can be seen as intra-cell pilot contamination.} Specifically, we propose PR among UTs, where the required number of orthogonal pilots can be much smaller than the number of UTs in a cell, by allowing that different UTs in a cell share the same pilot. We investigate pilot scheduling, channel estimation, and robust UL and DL transmissions under the systematic PR framework, which, has not been thoroughly addressed in the literature yet. Compared with existing works which dealt with inter-cell pilot contamination, our work is more general in the sense that we deal with both the pilot interference and pilot overhead under the PR framework. In this paper, our analysis focuses on the single-cell scenario for the sake of clarity.  Regarding the multi-cell scenario where pilot contamination persists \cite{Ngo13multicell}, performances of the proposed PR scheme will depend on the system configurations. For example, if statistical CSI coordination among cells is possible, i.e., channel covariance matrices can be exchanged among cells, pilot scheduling and pilot interference (as well as pilot contamination) mitigation can be performed jointly via exploiting statistical CSI of the intra-cell as well as inter-cell UTs.

\emph{Notations:} We use $\barjmath=\sqrt{-1}$ to denote the imaginary unit. Upper (lower) case boldface letters are used to denote matrices (column vectors). ${\bI}_{N}$ denotes the $N\times N$ dimensional identity matrix, and the subscript is omitted for brevity in some cases where it is clear. $\bzero$ denotes the all-zero vector (matrix). The superscripts $(\cdot)^{H}$, $(\cdot)^{T}$, and $(\cdot)^{*}$ denote the conjugated-transpose, transpose, and conjugate operations, respectively. The operator $\diag{\bx}$ denotes the diagonal matrix with $\bx$ along its main diagonal, and $\tr{\cdot}$ denotes the matrix trace operation. We employ $\vecele{\ba}{i}$ and $\vecele{\bA}{i,j}$ to denote the \ith{i} element of the vector $\ba$, and the $(i,j)^{\textrm{th}}$ element of the matrix $\bA$, respectively, where the element indices start from $1$. $\normt{\ba}=\sqrt{\ba^{H}\ba}$ denotes the $\ell_2$-norm of $\ba$, and $\normf{\bX}=\sqrt{\tr{\bX^{H}\bX}}$ denotes the Frobenius norm of $\bX$. $\inpro{\ba}{\bb}=\ba^{H}\bb$ denotes the inner product between $\ba$ and $\bb$. $\bA\succ\bzero$ $(\bA\succeq\bzero)$ denotes that $\bA$ is Hermitian positive definite (semi-definite), and $\bA\succeq\bB$ means that $\bA-\bB\succeq\bzero$. $\mathds{C}^{M\times N}$ denotes the $M\times N$ dimensional complex vector space. $\expect{\cdot}$ denotes the expectation operation. $\GCN{\ba}{\bB}$ denotes the circular symmetric complex Gaussian distribution with mean $\ba$ and covariance $\bB$. $\delta(\cdot)$ denotes the Dirac delta function. $\lfloor x\rfloor$ denotes the largest integer that is not greater than $x$. The notation $\triangleq$ is used for definitions,  and $\sim$ means ``be distributed as''. The superscripts ``p'', ``u'', and ``d'' stand for the expressions related to pilot, UL data, and DL data, respectively.

\emph{Outline:} The rest of the paper is organized as follows. In \secref{sec:massive_channel}, we investigate the massive MIMO channel model, and establish the relationship between channel spatial correlations and channel PAS. In \secref{sec:channel_train}, we present PR for UL channel estimation and show how PR affects the channel estimation performance. We also provide a condition under which the MSE-CE can be minimized. In \secref{sec:robust tra}, we develop the robust multiuser UL receiver and DL precoder under the MMSE-SD criterion, and reveal a MMSE duality between them. In \secref{sec:pilot_sche}, we study pilot scheduling and propose a low complexity pilot scheduling algorithm that relies on the channel statistics only. Simulation results are provided in \secref{sec:simulation} and the paper is concluded in \secref{sec:conc}.

\section{Massive MIMO Channel Model}\label{sec:massive_channel}

We consider massive MIMO transmission in the TDD mode in single-cell scenario, where the BS with $M$ antennas serves $K(\ll M)$ single-antenna UTs over frequency-flat fading channels on a narrow-band sub-carrier. We assume that channels vary in time according to the block fading model, where channel states stay constant over the coherence block with a length of $T$ symbols, and evolve from block to block in an independent and identically distributed manner according to some ergodic process.

With the ray-tracing based approach \cite{Tse05Fundamentals,Liu03Capacity,Barriac04Characterizing}, the UL channel between the $M$ antennas at the BS and the antenna of the \ith{k} UT can be modeled as
\begin{equation}\label{eq:mpath_model}
  \bg_{k}=\int_{\cA}\!\vtheta g_{k}\left(\theta\right) \intdx{\theta}=\int\limits_{\thetamin}^{\thetamax}\!\vtheta g_{k}\left(\theta\right) \intdx{\theta}
\end{equation}
where $g_{k}\left(\theta\right)$ and $\vtheta\in\mathds{C}^{M\times 1}$ are the complex channel gain function and the BS array response vector corresponding to the incidence angle $\theta$, respectively. We assume that $\normt{\vtheta}=\sqrt{M}$ for power normalization. We assume that the channel power seen at the BS is constrained to lie in the angle interval $\cA=[\thetamin,\thetamax]$, which can be achieved via placing directional antennas at the BS, and thus no power is received at the BS for incidence angle $\theta\notin\cA$.

We assume that the channel phases are uniformly distributed, thus $\expect{\bg_{k}}=\bzero$. We assume that channels with different incidence angles are uncorrelated, i.e.,
$\expect{g_{k}\left(\theta\right)g_{k}^{*}\left(\theta'\right)}=\beta_{k}\pktheta\delfunc{\theta-\theta'}$
where $\beta_{k}$ represents the large scale fading, and $\pktheta$ represents the channel PAS which models the channel power distribution in the angular domain \cite{Pedersen00stochastic}. Then from \eqref{eq:mpath_model}, the channel covariance matrix (BS spatial correlation matrix) is given by
\begin{align}\label{eq:def_chan_cov}
  \bR_{k}=\expect{\bg_{k}\bg_{k}^{H}}
  =\beta_{k}\int\limits_{\thetamin}^{\thetamax}\!{\vtheta\bv^{H}\left(\theta\right)\pktheta\intdx{\theta}}.
\end{align}

We assume that $\int_{-\infty}^{\infty}{\pktheta\intdx{\theta}}=1$, and channel power normalization should be satisfied as
\begin{equation}\label{eq:trrk}
  \tr{\bR_{k}}=\beta_{k}M\int\limits_{\thetamin}^{\thetamax}\!{\pktheta\intdx{\theta}}.
\end{equation}

A specific property of the massive antenna array is its high resolution to the channels in the angular domain \cite{Gao12Measured}, and we introduce an assumption about it in the following.

\begin{ass}\label{ass:asymp_ortho_asv}
\emph{\cite{Viberg95Performance}}
Array response vectors corresponding to distinct angles are asymptotically orthogonal when the BS antenna number tends to infinity, i.e., for $\forall \zeta,\vartheta\in\cA$,
\begin{equation}\label{eq:asymp_ortho_asv}
   \liminfty{M}\oneon{M}\inpro{\vx{\zeta}}{\vx{\vartheta}}=\delfunc{\zeta-\vartheta}.
\end{equation}
\end{ass}

Note that \asref{ass:asymp_ortho_asv} is valid for uniform linear array (ULA) as one shall see in \remref{rem:ula_dft_eig}. Based on this assumption, we can obtain the following result on massive MIMO channel covariance matrix.

\begin{lemma}\label{lm:Decomp_cov}
Let
\begin{equation}\label{eq:cov_eigv_mat}
  \bV=\oneon{\sqrt{M}}\left[\vx{\vartheta\left(\psi_{0}\right)},\vx{\vartheta\left(\psi_{1}\right)},\ldots,\vx{\vartheta\left(\psi_{M-1}\right)}\right]
\end{equation}
\begin{align}\label{eq:cov_eigen_ele}
  \vecele{\br_{k}}{m}=\beta_{k}M\cdot \mathsf{S}_{k}\left(\vartheta\left(\psi_{m-1}\right)\right)\left[\vartheta\left(\psi_{m}\right)-\vartheta\left(\psi_{m-1}\right)\right],\ntb
  \quad\textrm{for}\quad m=\oneldots{M}
\end{align}
where $\psi_{m'}=m'/M$ for $m'=0,1,\ldots,M$, and $\theta=\vartheta\left(\psi\right)$ over the support $\left[0,1\right]$ is a strictly increasing continuous function\footnote{The function $\vartheta\left(\psi\right)$ can be interpreted as a mapping from the space domain to the physical angle domain, and it indeed depends on the BS array structure. We assume the function $\vartheta\left(\psi\right)$ to be strictly increasing and continuous over the support to guarantee that the function is a one-to-one mapping.} that satisfies $\vartheta\left(0\right)=\thetamin$ and $\vartheta\left(1\right)=\thetamax$. Then under \asref{ass:asymp_ortho_asv}, matrices $\bV^{H}\bV$ and $\bR_{k}$ tend to be the identity matrix and $\bV\diag{\br_{k}}\bV^{H}$, respectively, when $M\to\infty$, in the sense that, for fixed positive integers $i$ and $j$,
\begin{align}\label{eq:cov_eig_uni}
  \liminfty{M}\vecele{\bV^{H}\bV-\bI_{M}}{i,j}=0
\end{align}
\begin{align}\label{eq:Decomp_cov}
  \liminfty{M}\vecele{\bR_{k}-\bV\diag{\br_{k}}\bV^{H}}{i,j}=0.
\end{align}
\end{lemma}
\begin{proof}
See \appref{app:lemma_cov_decomp}.
\end{proof}

The result in \lmref{lm:Decomp_cov} indicates that, when the BS antenna number $M$ is sufficiently large, the channel covariance matrix $\bR_{k}$ can be well approximated by
\begin{equation}\label{eq:cov_appr}
  \bR_{k}\approx\bV\diag{\br_{k}}\bV^{H}.
\end{equation}
Note that the matrix $\bV$ tends to be unitary when $M$ is sufficiently large. This establishes a relationship between channel spatial correlations and channel power distribution in the angular domain. Specifically, for massive MIMO channels, eigenvector matrices of the channel covariance matrices for different UTs tend to be the same, and are determined by the BS array response vectors, while eigenvalues depend on the respective channel PASs.

\begin{remark}\label{rem:ula_dft_eig}
When BS is equipped with the ULA, and the $M$ antennas are spaced with a half wavelength distance, the array response vector can be represented as \cite{Tse05Fundamentals}
\begin{align}\label{eq:ula steer vec}
  \vtheta&=\bigg[1,\expx{-\barjmath\pi\sintheta},\ldots,\ntb
  &\qquad\qquad\expx{-\barjmath\pi(M-1)\sintheta}\bigg]^{T}.
\end{align}
We assume that the AoA interval equals $\cA=[-\pi/2,\pi/2]$, and it is not hard to show \eqref{eq:asymp_ortho_asv} in \asref{ass:asymp_ortho_asv}, i.e., \asref{ass:asymp_ortho_asv} is valid in this case. Let $\theta=\vartheta\left(\psi\right)=\arcsin\left(2\psi-1\right)$, then $\vartheta\left(\psi_{m'}\right)=\arcsin\left(\frac{2m'}{M}-1\right)$ for $m'=0,1,\ldots,M$, and elements of $\bV$ reduce to $\vecele{\bV}{i,j}=\oneon{\sqrt{M}}\expx{-\barjmath2\pi\frac{\left(i-1\right)\left(j-1-M/2\right)}{M}}$ for $i=\oneldots{M}$ and $j=\oneldots{M}$. This indicates that, for the ULA case, when $M$ is sufficiently large, eigenvector matrix of the channel covariance matrix can be well approximated by the unitary discrete Fourier transform (DFT) matrix (up to some matrix elementary operations). Similar channel covariance matrix decomposition for the ULA case was derived in \cite{Yin13coordinated} and \cite{Adhikary13Joint}, however, the result in \lmref{lm:Decomp_cov} applies to the more general BS array configurations. In addition, the relationship between eigenvalues of the channel covariance matrix and channel PAS is established in \lmref{lm:Decomp_cov}.
\qedend
\end{remark}

The channel model proposed in \lmref{lm:Decomp_cov}, as well as the ULA case in \remref{rem:ula_dft_eig}, is based on \asref{ass:asymp_ortho_asv}, where angular resolution of the antenna array is assumed to tend to infinity when the antenna number grows to infinity. It is well known that angular resolution of an array is proportional to the array size \cite{Tse05Fundamentals,Clerckx13MIMO}. Thus, the channel model in \lmref{lm:Decomp_cov} is applicable for arbitrary antenna array configuration with sufficiently large array size. However, in practical wireless communication systems, the antenna array size is always finite. Nevertheless, for a fixed array size, a fairly large number of antennas can still be accommodated at the BS if wireless transmission is performed over higher carrier frequency in, e.g., millimeter wave massive MIMO systems \cite{Swindlehurst14Millimeter}. For example, considering the case where the antenna number equals $128$ and the carrier frequency equals $30$ $\mathrm{GHz}$, which lies in the millimeter wave spectrum, the size of the ULA with half wavelength spacing considered in \remref{rem:ula_dft_eig} is just $0.64$ $\mathrm{m}$. Note that the channel model in \lmref{lm:Decomp_cov} for the ULA case has been shown as a good approximation with finite but large number of antennas \cite{Yin13coordinated,Adhikary13Joint}. For these reasons, the proposed channel model is of great importance from both practical and theoretical perspectives.

In this paper, we employ the widely accepted assumption that channels are wide-sense stationary \cite{Clerckx13MIMO}, thus channel covariance matrices can be obtained by the BS. However, stationarity of the realistic wireless channels can only be satisfied in a local manner, i.e., channel covariance matrices also vary over time. Thus, it requires that the channel covariance matrices being periodically estimated at the BS. Estimation of massive MIMO channel covariance matrices is rather challenging and resource-consuming \cite{Marzetta11random}. However, from the result given by Lemma 1, only the eigenvalues rather than the whole massive MIMO channel covariance matrices need to be estimated, thus the number of parameters to estimate can be significantly reduced. In addition, the channel covariance matrices vary much less frequently than the instantaneous CSI, and thus can be estimated via averaging over time. Furthermore, the channel covariance matrices have been shown to stay constant over a wide frequency interval \cite{Barriac04Space}, and thus can be estimated via averaging over frequency in practical wideband systems. Therefore, there will be enough time-frequency resources to estimate the channel covariance matrices, and the estimation accuracy can be guaranteed in practice. In the rest of the paper, we will assume that the channel covariance matrices of all the UTs are known by the BS.

We assume that the channel elements to be jointly Gaussian from the law of large numbers, i.e., $\bg_{k}\sim\GCN{\bzero}{\bR_{k}}$. We assume that channels of different UTs are mutually statistically independent, and denote the UL channels of all the UTs as $\bG=[\bg_{1},\ldots,\bg_{K}]\in\mathds{C}^{M\times K}$.

\section{PR for UL Channel Training}\label{sec:channel_train}

In this section, we present PR for UL channel training, and investigate how PR affects the channel estimation performance. Our following analysis applies to arbitrary PR pattern, while how to form the PR pattern exploiting the statistical CSI will be discussed in \secref{sec:pilot_sche}.

We denote the UT set as $\cK=\{\oneldots{K}\}$ where $k\in\cK$ is the UT index. We assume that the UL training interval length equals $\tau(<K)$,\footnote{It should be noted that the results obtained in the following are applicable for arbitrary $\tau$, and $\tau$ can be either set to be a fixed number, or determined dynamically by the BS. One example of how to dynamically determine $\tau$ based on the net spectral efficiency maximization criterion will be discussed in \secref{subsec:net_rate}.} and all the UTs transmit the respective pilot sequences in the length of $\tau$ simultaneously during the training interval.\footnote{For the case that one UT sends pilot signals during several particular channel uses while remains silent during other channel uses, the pilot signals can be seen as a specific pilot sequence with several non-zero entries at the corresponding channel uses.} Note that the maximum number of available orthogonal sequences is equal to the sequence length, and we assume that the available orthogonal pilot sequence number equals $\tau$ for simplicity. We denote the available orthogonal pilot set as $\cT=\{\oneldots{\tau}\}$, and the \ith{\pi} pilot sequence as $\bx_{\pi}\in\mathds{C}^{\tau\times1}$ where $\pi\in\cT$ is the orthogonal pilot index. Different pilot sequences are assumed to satisfy the orthogonal condition that $\bx_{\pi}^{H}\bx_{\pi'}=\tau\dnnot{\upnot{\sigma}{p}}{x}\cdot\delfunc{\pi-\pi'}$ where $\dnnot{\upnot{\sigma}{p}}{x}$ is the pilot signal transmit power.

We denote an arbitrary PR pattern with UT set $\cK$ and pilot set $\cT$ as $\pktau=\{(k,\Pindex{k}):k\in\cK,\Pindex{k}\in\cT\}$ where $(k,\Pindex{k})\in\pktau$ denotes that the \ith{\pi_{k}} pilot sequence $\bx_{\Pindex{k}}$ is allocated to the \ith{k} UT. We use $\Kset{\pi}=\{k:\Pindex{k}=\pi\}$ to denote the set of the UTs using the \ith{\pi} pilot sequence.

With the PR pattern $\pktau$, the UTs transmit their assigned pilots periodically to enable the BS to estimate the channels. During the UL training phase of each coherence block, the received pilot signals at the BS can be written as
\begin{equation}
  \bY=\bG\bX+\bN\in\mathds{C}^{M\times\tau}
\end{equation}
where $\bG$ is the UL channel matrix, $\bX=\left[\bx_{\pi_{1}},\bx_{\pi_{2}},\ldots,\bx_{\pi_{K}}\right]^{T}\in\mathds{C}^{K\times\tau}$ is the UL pilot signal matrix, $\bN$ is the independent additive Gaussian noise matrix with elements distributed as independently and identically $\GCN{0}{\dnnot{\upnot{\sigma}{p}}{z}}$, and $\dnnot{\upnot{\sigma}{p}}{z}$ is the noise power during the training phase. After decorrelation and power normalization of the received signals \cite{Marzetta10Noncooperative}, the BS can obtain the channel observation of all the UTs. Specifically, for the \ith{k} UT in a given coherence block, the BS obtains the UL channel observation as
\begin{align}\label{eq:lses}
    \ypkup&=\oneon{\dnnot{\upnot{\sigma}{p}}{x}\tau}\bY\bx_{\pi_{k}}^{*}
    =\oneon{\dnnot{\upnot{\sigma}{p}}{x}\tau}
    \left(\sum_{\ell=1}^{K}\bg_{\ell}\bx_{\pi_{\ell}}^{T}+\bN\right)\bx_{\pi_{k}}^{*}\ntb
    &=\sum_{\ell=1}^{K}\bg_{\ell}\cdot\delfunc{\pi_{\ell}-\pi_{k}}
    +\oneon{\dnnot{\upnot{\sigma}{p}}{x}\tau}\bN\bx_{\pi_{k}}^{*}\ntb
    &=\sum_{\ell\in\KPset{k}}\bg_{\ell}
    +\oneon{\dnnot{\upnot{\sigma}{p}}{x}\tau}\bN\bx_{\pi_{k}}^{*}.
\end{align}

With the property of unitary transformation, it is not hard to show that the noise term $\oneon{\dnnot{\upnot{\sigma}{p}}{x}\tau}\bN\bx_{\pi_{k}}^{*}$ in \eqref{eq:lses} is still Gaussian with elements distributed as independently and identically $\GCN{0}{\frac{\dnnot{\upnot{\sigma}{p}}{z}}{\dnnot{\upnot{\sigma}{p}}{x}\tau}}$. Let $\upnot{\rho}{p}=\dnnot{\upnot{\sigma}{p}}{x}/\dnnot{\upnot{\sigma}{p}}{z}$ be the UL channel training signal-to-noise ratio (SNR), then \eqref{eq:lses} can be rewritten as
\begin{align}\label{eq:ypkup}
    \ypkup&=\sum_{\ell\in\KPset{k}}{\bg_{\ell}}+\invsqrttaurhoup\npkup\ntb
    &=\bg_{k}+\underbrace{\sum_{\ell\in\setdif{\KPset{k}}{\eletoset{k}}}{\bg_{\ell}}}_{\textrm{pilot interference}}
    +\underbrace{\invsqrttaurhoup\npkup}_{\textrm{pilot noise}}
\end{align}
where ``$\backslash$'' denotes the set subtraction operation, and $\npkup\sim \GCN{\bzero}{\bI_{M}}$ is the normalized additive noise. Note that $\KPset{k}$ represents the set of the UTs using the same pilot as the \ith{k} UT, and the BS has to estimate the channels of all the UTs reusing the \ith{\Pindex{k}} pilot based on the observation $\ypkup$. The MMSE estimate of the channel $\bg_{k}$ based on the channel observation $\ypkup$ is given by
\begin{equation}\label{eq:est_cha}
    \hatbg_{k}=\bR_{k}\invcovpi{k}\ypkup
\end{equation}
where
\begin{equation}\label{eq:cov_obs}
    \covpi{k}\triangleq\sum_{\ell\in\KPset{k}}{\bR_{\ell}}+\invtaurhoup\bI.
\end{equation}
From the orthogonality principle of MMSE estimation \cite{Kailath00Linear}, channel estimation error $\tildebg_{k}=\bg_{k}-\hatbg_{k}$ is independent of $\hatbg_{k}$, and the covariance of $\tildebg_{k}$ is
\begin{equation}\label{eq:cov_error}
  \bR_{\tildebg_{k}}=\bR_{k}-\bR_{k}\invcovpi{k}\bR_{k}.
\end{equation}
Note that $\hatbg_{k}$ and $\bR_{\tildebg_{k}}$ are also mean and covariance of $\bg_{k}$ conditioned on $\ypkup$, respectively \cite{Kailath00Linear}.

The estimation error covariance is an important measure of the estimation performance, and we define the MSE-CE as
\begin{equation}\label{eq:mse_est_obj}
  \upup{\epsilon}\triangleq\sum_{k=1}^{K}\tr{\bR_{\tildebg_{k}}}.
\end{equation}

Before we proceed, we first define the orthogonality between two arbitrary Hermitian positive semi-definite matrices using the angle ($0\leq\theta\leq\pi/2$) between them as
\begin{align}\label{eq:thcosdef}
  \thcos{\bA}{\bB}\triangleq\arccos\frac{\tr{\bA^{H}\bB}}{\normf{\bA}\normf{\bB}}
  &=\arccos\frac{\tr{\bA\bB}}{\normf{\bA}\normf{\bB}},\ntb
  &\textrm{for}\quad \bA,\bB\succeq\bzero.
\end{align}
Then we present a condition under which the MSE-CE defined in \eqref{eq:mse_est_obj} can be minimized in the following theorem.
\begin{theorem}\label{th:est_min_con}
The minimum value of the MSE-CE $\upup{\epsilon}$ is given by
\begin{equation}\label{eq:cha_sum_mmse}
  \upup{\varepsilon}=\sum_{k=1}^{K}\tr{\bR_{k}-\bR_{k}\inv{\left(\bR_{k}+\invtaurhoup\bI\right)}\bR_{k}}
\end{equation}
and the minimum is achieved under the condition that, for $\forall i,j\in\cK$ and $i\neq j$,
\begin{equation}
  \thcos{\bR_{i}}{\bR_{j}}=\frac{\pi}{2},\qquad\textrm{when} \quad \Pindex{i}=\Pindex{j}.
\end{equation}
\end{theorem}
\begin{proof}
See \appref{app:est_min_con}.
\end{proof}

In the MSE-CE metric defined in \eqref{eq:mse_est_obj}, correlations between the channel estimation errors seen by different UTs are not taken into account. Actually, the correlations between the channel estimation errors of the UT $i$ and the UT $j\left(j\neq i\right)$ can be obtained as
\begin{align}
\expect{\tildebg_{i}\tildebg_{j}^{H}}
&=\expect{\left(\bg_{i}-\bR_{i}\bC_{\pi_{i}}^{-1}\upup{\by_{\pi_{i}}}\right)
\left(\bg_{j}-\bR_{j}\bC_{\pi_{j}}^{-1}\upup{\by_{\pi_{j}}}\right)^{H}}\ntb
&=-\bR_{i}\bC_{\pi_{i}}^{-1}\bR_{j}\cdot\delfunc{\pi_{i}-\pi_{j}}
\end{align}
which indicates that channel estimation errors of the UTs with orthogonal pilots are independent, while those of the UTs reusing the pilots are correlated. However, if the condition given in \thref{th:est_min_con} is satisfied, then $-\bR_{i}\bC_{\pi_{i}}^{-1}\bR_{j}=\bzero$, i.e., channel estimation errors seen by different UTs will be uncorrelated no matter whether they reuse the pilots or not, and the condition given in \thref{th:est_min_con} is still optimal.

To obtain clear insights of \thref{th:est_min_con}, we consider the asymptotic antenna number case, and the following corollary can be readily obtained from \lmref{lm:Decomp_cov}.
\begin{corol}\label{co:asymp_est_min_con}
When the BS antenna number $M\to\infty$, the MSE-CE $\upup{\epsilon}$ can be minimized provided that, for $\forall i,j\in\cK$ and $i\neq j$,
\begin{equation}
  \inpro{\br_{i}}{\br_{j}}=0,\qquad\textrm{when} \quad \Pindex{i}=\Pindex{j}
\end{equation}
where $\br_{i}$ for fixed positive integer $i$ is given in \lmref{lm:Decomp_cov}.
\qedend
\end{corol}

The result in \coref{co:asymp_est_min_con} indicates that the MSE-CE $\upup{\epsilon}$ can be minimized if the UTs reusing the pilots have non-overlapping channel AoA intervals. This result is very intuitive, as in such cases, the channels of different UTs are strictly separated in the angular domain, and the pilot interference does not take into effect. Moreover, in the high SNR regime where the training SNR $\rhoup\to\infty$, the pilot noise vanishes, and then the MSE-CE $\upup{\varepsilon}\to 0$, which implies that channel estimations tend to be perfect.

Although the conditions in \thref{th:est_min_con} and \coref{co:asymp_est_min_con} are desirable, they cannot always be well satisfied. However, in realistic outdoor propagation environments where the BS is located at an elevated position, channel AS seen by the BS is usually small \cite{Vaughan03Channels,Clerckx13MIMO}, which indicates that most of the channel power is concentrated in a narrow angle interval, and the channel power outside this angle interval is very small. For UTs located geographically apart in different spatial directions, the overlaps of their channel power in the angular domain might be neglected, and thus PR becomes feasible in such spatially correlated massive MIMO channels.

\section{Robust UL/DL Data Transmissions}\label{sec:robust tra}

In the previous section, we showed feasibility of PR for massive MIMO transmission, and presented UL channel training with PR. In each coherence block, the BS obtains the channel estimates of all the UTs after UL channel training. The conventional data transmission design in massive MIMO treats the channel estimates as the real channels. However, with PR, the channel estimation performance will degrade in most cases, thus a robust data transmission design with channel estimation errors taken into account is of paramount importance in the considered PR based massive MIMO transmission. There are two main approaches to design a wireless system robust to the channel uncertainty: the worst-case approach and the statistical approach. In the worst-case approach, the channel uncertainty is modeled as being within a given set around the channel estimate, and a worst-case transmission performance can be guaranteed \cite{Pascual06robust}. In the statistical approach, the channel uncertainty is modeled using the channel statistics, such as the mean and the covariance, and a statistical average performance can be guaranteed \cite{Zhang08Statistically}. In this work, we employ the statistical approach to model the channel uncertainty. Specifically, in each coherence block, based on the received pilot signals, the CSI uncertainty at the BS can be modeled statistically using its conditional distribution, i.e., the conditional mean (the MMSE channel estimate) and the conditional covariance (the covariance of the channel estimation error). Note that the channel estimation error covariance $\bR_{\tildebg_{k}}$ given in \eqref{eq:cov_error} depends on the PR pattern $\pktau$ and the channel covariance, and thus can be known by the BS. In the following, we will develop robust data transmissions for UL and DL, respectively, under the MMSE-SD criterion.

\subsection{Robust UL Data Transmission}

During the UL data transmission phase, the signal received at the BS at a channel use in the given coherence block can be expressed as
\begin{align}\label{eq:upuby}
  \upu{\by}=\bG\upu{\ba}+\invsqrtrhou\bnu
  =\left(\hatbG+\tildebG\right)\upu{\ba}+\invsqrtrhou\bnu
\end{align}
where $\hatbG=\left[\hatbg_{1},\hatbg_{2},\ldots,\hatbg_{K}\right]$ is the channel estimate, $\tildebG=[\tildebg_{1},\ldots,\tildebg_{K}]$ is the channel estimation error, $\upu{\ba}\in\mathds{C}^{K\times 1}$ with mean $\bzero$ and covariance $\bI_{K}$ denotes the UL data signal vector where $\vecele{\upu{\ba}}{k}$ is the signal sent by the \ith{k} UT, $\bnu\sim\GCN{\bzero}{\bI_{M}}$ is the independent additive noise, and $\rhou$ is the UL data transmission SNR per UT.

We consider the linear receiver at the BS
\begin{equation}
  \upu{\hatba}=\bW^{T}\upu{\by}
\end{equation}
and then MSE-SD of the UL transmission in the given coherence block can be defined as
\begin{align}\label{eq:mseu}
  \upu{\epsilon}\triangleq\expect{\sqrnormt{\upu{\hatba}-\upu{\ba}}}
\end{align}
where the expectation is with respect to $\upu{\ba}$, $\bnu$, and the channel estimation error $\tildebG$.

Finding the optimal UL receiver based on the MMSE-SD criterion can be formulated as
\begin{align}\label{eq:ul opt prob}
  \mini{\bW}\qquad\upu{\epsilon}
\end{align}
and we present the solution in the following theorem.

\begin{theorem}\label{th:ul bf}
  The optimal solution to the problem \eqref{eq:ul opt prob} is given by
  \begin{equation}\label{eq:rob_opt_rec}
    \upopt{\bW}=\left[\inv{\left(\hatbG\hatbG^{H}+\sum_{k=1}^{K}{\bR_{\tildebg_{k}}}+\invrhou\bI\right)}\hatbG\right]^{*}
  \end{equation}
  and the corresponding MSE-SD is given by
  \begin{equation}\label{eq:ul_mmse}
    \upumin{\epsilon}=\tr{\inv{\left(\bI+\hatbG^{H}\inv{\left(\sum_{k=1}^{K}{\bR_{\tildebg_{k}}}+\invrhou\bI\right)}\hatbG\right)}}.
  \end{equation}
\end{theorem}
\begin{proof}
See \appref{app:ul bf}.
\end{proof}

For the conventional receiver with channel estimates assumed to be accurate, the impact of the channel estimation error is omitted. While for our robust MMSE receiver design, the channel estimation error due to PR is taken into account. Specifically, the expectation in \eqref{eq:mseu} accounts for the channel estimation error $\tildebG$, which leads to our robust MMSE receiver. Note that the robust MMSE receiver given in \eqref{eq:rob_opt_rec} exhibits a similar structure to the conventional receiver. When $\sum_{n=1}^{K}{\bR_{\tildebg_{n}}}\to\bzero$, the robust MMSE receiver in \eqref{eq:rob_opt_rec} reduces to the conventional receiver
\begin{equation}\label{eq:conv_mmse_re}
  \bW^{\textrm{con}}=\left[\inv{\left(\hatbG\hatbG^{H}+\invrhou\bI\right)}\hatbG\right]^{*}.
\end{equation}

\subsection{Robust DL Data Transmission}

During the DL data transmission phase, the signal received at the UTs at a channel use in the given coherence block can be expressed as
\begin{align}\label{eq:DL signal}
    \upd{\by}=\bG^{T}\bB\upd{\ba}+\invsqrtrhod\bnd
    =\left(\hatbG+\tildebG\right)^{T}\bB\upd{\ba}+\invsqrtrhod\bnd
\end{align}
where the DL channel $\bG^{T}$ is the transpose of the UL channel due to the channel reciprocity of the TDD systems \cite{Marzetta10Noncooperative}, $\upd{\ba}\in\mathds{C}^{K\times 1}$ with mean $\bzero$ and covariance $\bI_{K}$ denotes the DL data signal vector where  $\vecele{\upd{\ba}}{k}$ is the signal for the \ith{k} UT, $\upd{\bn}\sim \GCN{\bzero}{\bI_{K}}$ is the independent additive noise, $\rhod$ is the average DL data transmission SNR per UT, and $\bB$ is the DL linear precoding matrix which satisfies the power constraint
\begin{equation}
    \tr{\bB\bB^{H}}\leq K.
\end{equation}
Then MSE-SD of the DL transmission in the given coherence block can be defined as
\begin{equation}\label{eq:msed}
  \upd{\epsilon}\triangleq\expect{\sqrnormt{\alpha\upd{\by}-\upd{\ba}}}
\end{equation}
where the expectation is with respect to $\upd{\ba}$, $\bnd$, and $\tildebG$, and $\alpha$ is a real scalar parameter corresponding to the potential power scaling performed at the UTs.\footnote{Note that $\alpha$ is a real scalar, and the overhead for the UTs to obtain it can be neglected.}

Finding the optimal DL precoder based on the MMSE-SD criterion can be formulated as
\begin{align}\label{eq:dl opt prob}
\begin{split}
    \mini{\bB,\alpha}\qquad&\upd{\epsilon}\\
    \st\qquad &\tr{\bB\bB^{H}}\leq K
\end{split}
\end{align}
and we present the solution in the following theorem.

\begin{theorem}\label{th:dl bf}
  The optimal solution to the problem \eqref{eq:dl opt prob} is given by
  \begin{equation}\label{eq:rob_opt_pre}
    \upopt{\bB}=\invgammaopt\left[\inv{\left(\hatbG\hatbG^{H}+\sum_{k=1}^{K}{\bR_{\tildebg_{k}}}+\invrhod\bI\right)}\hatbG\right]^{*}
  \end{equation}
    \begin{equation}\label{eq:alpopt}
    \alphaopt=\upopt{\gamma}
  \end{equation}
  where $\upopt{\gamma}$ is chosen to satisfy the power normalization constraint $\tr{\upopt{\bB}(\upopt{\bB})^{H}}=K$, its value is given by
  \begin{equation}\label{eq:power_norm_gamma}
    \upopt{\gamma}=\sqrt{\frac{\tr{\hatbG^{H}\sqrinv{\left(\hatbG\hatbG^{H}
    +\sum\limits_{k=1}^{K}{\bR_{\tildebg_{k}}}+\invrhod\bI\right)}\hatbG}}{K}}
  \end{equation}
  and the corresponding MSE-SD is given by
  \begin{equation}\label{eq:dl_mmse}
    \updmin{\epsilon}=\tr{\inv{\left(\bI+\hatbG^{H}\inv{\left(\sum_{k=1}^{K}{\bR_{\tildebg_{k}}}+\invrhod\bI\right)}\hatbG\right)}}.
  \end{equation}
\end{theorem}
\begin{proof}
See \appref{app:dl bf}.
\end{proof}

From \thref{th:dl bf}, we can observe that, similarly to the UL case, the robust MMSE precoder in \eqref{eq:rob_opt_pre} also embraces the structure of the conventional precoder.

\subsection{UL-DL Duality}

From the results in \thref{th:ul bf} and \thref{th:dl bf}, we can readily obtain the following UL-DL MMSE duality.
\begin{corol}\label{co:ul dl duality}
  In each TDD coherence block, if $\rhou=\rhod$, then $\upopt{\bB}=\upopt{\bW}/\upopt{\gamma}$, and $\upumin{\epsilon}=\updmin{\epsilon}$.
  \qedend
\end{corol}

The result in \coref{co:ul dl duality} indicates that, in the same TDD coherence block, if the UL data transmission SNR equals the DL data transmission SNR, then the robust DL precoder in \eqref{eq:rob_opt_pre} can be achieved by the robust UL receiver in \eqref{eq:rob_opt_rec} with proper power normalization, and the complexity of computing the robust MMSE DL precoder can be reduced. In addition, if the robust MMSE receiver and robust MMSE precoder are used for data transmissions, then the same MMSE-SD can be achieved in both the UL and DL in each TDD coherence block. Note that similar UL-DL duality based on the perfect CSI assumption was provided in the literature such as \cite{Viswanath03Sum} and \cite{Shi07Downlink}, however, our result in \coref{co:ul dl duality} is established on the pilot-assisted CSI acquisition assumption.

In this section, we have investigated the robust UL and DL data transmissions with channel estimation error due to PR taken into account. It will be seen in the following section that, PR based massive MIMO transmission, which combines pilot scheduling and robust data transmission, can achieve the MMSE-SD optimality for both the UL and DL.

\section{Pilot Scheduling}\label{sec:pilot_sche}

Up to now, we have investigated channel training and data transmission of the massive MIMO transmission with PR, and the obtained results are applicable to arbitrary PR pattern. In this section, we study pilot scheduling which exploits the long term statistical CSI to allocate the available pilot signals to the UTs, and we focus on two MMSE related criteria.

\subsection{MMSE-CE Criterion}

CSI is critical to massive MIMO transmission, and it is natural to design the pilot scheduler based on the MMSE-CE criterion, which leads to the following problem
\begin{align}\label{eq:pilot_sche_comb}
  \mini{\pktau}&\qquad\upup{\epsilon}
\end{align}
where $\upup{\epsilon}$ is defined in \eqref{eq:mse_est_obj}.

The pilot scheduling problem in \eqref{eq:pilot_sche_comb} is combinatorial, and the optimal PR pattern $\pktau$ can be found through exhaustive search (ES). The complexity of the ES in \eqref{eq:pilot_sche_comb}, in terms of the (complex) scalar multiplication number which dominates the computational complexity, is briefly evaluated as follows. Recalling \eqref{eq:mse_est_obj}, the scalar multiplication number required in evaluation of the objective function in \eqref{eq:pilot_sche_comb} is $\Order{M^{3}K}$. Thus, the computational complexity of running ES under the MMSE-CE criterion is $\Order{\tau^{K}M^{3}K}$.

\subsection{MMSE-SD Criterion}

MSE-SD is an important performance measure of the data transmission, and in the sequel we study the pilot scheduler design regarding the MSE-SD metric. Due to the UL-DL MMSE-SD duality in each coherence block presented in \coref{co:ul dl duality}, we assume that $\rhou=\rhod$ for simplicity, and denote that $\rhou=\rhod=\rhot$ and $\upumin{\epsilon}=\updmin{\epsilon}=\uptmin{\epsilon}$, where the superscript ``t'' stands for expression related to data transmission. We consider pilot scheduling under the MMSE-SD criterion, which can be formulated as
\begin{align}\label{eq:data opt prob}
  \mini{\pktau}\quad&\expect{\uptmin{\epsilon}}\ntb
  &\quad=\expect{\tr{\inv{\left(\bI+\hatbG^{H}\inv{\left(\uptneff{\bR}\right)}\hatbG\right)}}}
\end{align}
where the expectation is with respect to the channel fading and the noise distributions, and the effective noise covariance matrix is defined as
\begin{align}
  \uptneff{\bR}\triangleq\sum_{k=1}^{K}{\bR_{\tildebg_{k}}}+\invrhot\bI.
\end{align}
The objective function in \eqref{eq:data opt prob} is the average of the MMSE-SD that can be achieved by the robust MMSE receiver and robust MMSE precoder in each coherence block, and it depends on the statistics of the channel fading and pilot noise distributions. It should be noted that here we still use the term MMSE-SD for brevity, however the meaning of it differs from that when we consider the designs of the receiver and precoder in the previous section.

Due to the difficulty in obtaining the closed-form expression of the objective function $\expect{\uptmin{\epsilon}}$ in \eqref{eq:data opt prob}, we first present a lower bound of it in the following lemma.

\begin{lemma}\label{lm:ammse lb}
The average MSE-SD $\expect{\uptmin{\epsilon}}$ is lower bounded by
\begin{equation}\label{eq:average_mmse_lb}
  \expect{\uptmin{\epsilon}}\geq\upnot{\epsilon}{t,alb}=\tr{\inv{\left(\bI_{K}+\bOmega\right)}}
\end{equation}
where for fixed positive integers $i$ and $j$,
\begin{equation}\label{eq:omegaij}
  \vecele{\bOmega}{i,j}=\tr{\invcovpi{i}\bR_{i}\inv{\left(\uptneff{\bR}\right)}\bR_{j}}\cdot\delfunc{\Pindex{i}-\Pindex{j}}.
\end{equation}
\end{lemma}
\begin{proof}
See \appref{app:ammse lb}.
\end{proof}

It will be seen in \secref{sec:simulation-robust} that, the lower bound presented in \lmref{lm:ammse lb} is tight over a wide SNR region. By replacing the objective function $\expect{\uptmin{\epsilon}}$ with its lower bound presented in \lmref{lm:ammse lb}, the pilot scheduling problem \eqref{eq:data opt prob} can be simplified as
\begin{align}\label{eq:simpl data prob}
  \mini{\pktau}&\qquad\upnot{\epsilon}{t,alb}.
\end{align}
The pilot scheduling problem in \eqref{eq:simpl data prob} is also combinatorial. The optimal PR pattern $\pktau$ can be found through ES. Note that the scalar multiplication number required in evaluation of the objective function in \eqref{eq:simpl data prob} is $\Order{M^{3}K^{2}}$, thus, the computational complexity of running ES under the MMSE-SD criterion is $\Order{\tau^{K}M^{3}K^{2}}$.

Before we proceed, we present a condition under which $\upnot{\epsilon}{t,alb}$ can be minimized in the following theorem.

\begin{theorem}\label{th:pilot reuse}
The minimum value of the lower bound average MSE-SD $\upnot{\epsilon}{t,alb}$ is given by
\begin{equation}
  \upt{\varepsilon}=\sum_{i=1}^{K}\oneon{1+\vecele{\bomega}{i}}
\end{equation}
where $\vecele{\bomega}{i}$ for fixed positive integer $i$ is given by \eqref{eq:bomegai}, shown at the top of the next page, and the minimum is achieved under the condition that, for $\forall i,j\in\cK$ and $i\neq j$,
\begin{figure*}
\begin{align}\label{eq:bomegai}
  \vecele{\bomega}{i}=
  \tr{\inv{\left(\bR_{i}+\invtaurhoup\bI\right)}\bR_{i}\inv{\left[\sum_{k=1}^{K}\left(\bR_{k}-\bR_{k}\inv{\left(\bR_{k}+\invtaurhoup\bI\right)}\bR_{k}\right)+\invrhot\bI\right]}\bR_{i}}
\end{align}
\hrulefill
\end{figure*}
\begin{equation}
  \thcos{\bR_{i}}{\bR_{j}}=\frac{\pi}{2},\qquad\textrm{when} \quad \Pindex{i}=\Pindex{j}.
\end{equation}
\end{theorem}
\begin{proof}
See \appref{app:pilot reuse}.
\end{proof}

Recalling \lmref{lm:Decomp_cov}, we can readily obtain the following corollary.
\begin{corol}\label{co:asymp_tra_min_con}
When the BS antenna number $M\to\infty$, the lower bound average MSE-SD $\upnot{\epsilon}{t,alb}$ can be minimized provided that, for $\forall i,j\in\cK$ and $i\neq j$,
\begin{equation}
  \inpro{\br_{i}}{\br_{j}}=0,\qquad\textrm{when} \quad \Pindex{i}=\Pindex{j}
\end{equation}
where $\br_{i}$ for fixed positive integer $i$ is given in \lmref{lm:Decomp_cov}.
\qedend
\end{corol}

Interestingly, conditions for optimal data transmission obtained in \thref{th:pilot reuse} and \coref{co:asymp_tra_min_con} are the same as those for optimal channel training obtained in \thref{th:est_min_con} and \coref{co:asymp_est_min_con}. The intuitive interpretation lies in that, for massive MIMO transmission, if channels of the UTs reusing the pilots can be rigorously spatially separated, then not only the pilot interference but also the transmission data interference vanishes. Furthermore, in the high SNR regime where both the training SNR $\rhoup$ and transmission SNR $\rhot$ tend to infinity, the remaining additive noise vanishes, and the average MSE-SD $\upt{\varepsilon}\to 0$. This result shows that the PR based transmission scheme, which combines pilot scheduling and robust data transmission, can achieve the MMSE-SD optimality.

\subsection{SGPS Algorithm}\label{subsec:sgps}

In the above subsections, we have investigated pilot scheduling under two MMSE related criteria. In both cases, the designs are formulated as combinatorial optimization problems, and the optimal PR patterns can be formed through ES. However, due to the exponential complexity, ES becomes hard to implement in practice as the UT number grows.

In this subsection, we propose a low complexity pilot scheduling algorithm called the statistical greedy pilot scheduling (SGPS) algorithm which is motivated by the conditions for optimal channel estimation and data transmission given in \thref{th:est_min_con} and \thref{th:pilot reuse}, and the main idea is that channel covariance matrices of the UTs reusing the pilots should be as orthogonal as possible. Detailed description of the SGPS algorithm is summarized in \alref{alg:SSPS}. Coordinated pilot allocation algorithm for mitigating the inter-cell pilot contamination using similar idea was proposed in \cite{Yin13coordinated}, however, the above SGPS algorithm is dedicated for the single-cell scenario.

\begin{algorithm}[!t]
\caption{Statistical Greedy Pilot Scheduling (SGPS) Algorithm}
\label{alg:SSPS}
\begin{algorithmic}[1]
\Require
The UT set $\cK=\{\oneldots{K}\}$ and the channel covariance information $\bR_{k}(k\in\cK)$,
the orthogonal pilot set $\cT$ with the pilot length $\tau(1<\tau<K)$
\Ensure
PR pattern $\pktau=\{(k,\Pindex{k}):k\in\cK,\Pindex{k}\in\cT\}$
\State Initialize the unscheduled UT set $\upnot{\cK}{un}=\cK$, the unused pilot set $\upnot{\cT}{un}=\cT$
\Statex \textbf{Step 1)} Schedule the UTs with ``similar'' channel covariance matrices and assign them with orthogonal pilots
\State $m_{1}=1$, $\Pindex{1}=1$, $\Kset{1}=\eletoset{1}$, $\upnot{\cK}{un}\gets\setdif{\upnot{\cK}{un}}{\eletoset{1}}$, $\upnot{\cT}{un}\gets\setdif{\upnot{\cT}{un}}{\eletoset{1}}$
\While{$\upnot{\cT}{un}\neq\varnothing$}
\State For the pilot $t\in\upnot{\cT}{un}$, select the UT $m_{t}=\argmax{\ell\in\upnot{\cK}{un}}\ \sum_{j\in\setdif{\cT}{\upnot{\cT}{un}}}\costh{\bR_{\ell}}{\bR_{m_{j}}}$
\State Assign the pilot $t$ to the UT $m_{t}$, $\Pindex{m_{t}}=t$, $\Kset{t}=\eletoset{m_{t}}$
\State Update $\upnot{\cK}{un}\gets\setdif{\upnot{\cK}{un}}{\eletoset{m_{t}}}$, $\upnot{\cT}{un}\gets\setdif{\upnot{\cT}{un}}{\eletoset{t}}$
\EndWhile
\Statex \textbf{Step 2)} Each unscheduled UT is assigned with the ``best'' pilot so that the channel covariance matrices of the UTs reusing the pilots are as orthogonal as possible
\While{$\upnot{\cK}{un}\neq\varnothing$}
\State For the UT $k\in\upnot{\cK}{un}$, select the pilot $n_{k}=\argmin{q\in\cT}\ \sum_{s\in\Kset{q}}\costh{\bR_{k}}{\bR_{s}}$
\State Assign the pilot $n_{k}$ to the UT $k$, $\Pindex{k}=n_{k}$, $\Kset{n_{k}}\gets\Kset{n_{k}}\cup\eletoset{k}$
\State Update $\upnot{\cK}{un}\gets\setdif{\upnot{\cK}{un}}{\eletoset{k}}$
\EndWhile
\end{algorithmic}
\end{algorithm}

We evaluate the complexity of the SGPS algorithm as follows. In the process of the SGPS algorithm, no more than $\sum_{m=1}^{K-1}m(K-m)=(K-1)K(K+1)/6$ orthogonality calculations defined in \eqref{eq:thcosdef} are needed. Note that the scalar multiplication number needed in each orthogonality calculation is $\Order{M^{2}}$, thus, the computational complexity of running the SGPS algorithm is $\Order{M^{2}K^{3}}$. In above subsections, we have shown that the ES complexity under the MMSE-CE and MMSE-SD criteria are $\Order{\tau^{K}M^{3}K}$ and $\Order{\tau^{K}M^{3}K^{2}}$, respectively. This indicates that the SGPS algorithm gives a significant computational complexity reduction compared with ES. Meanwhile, simulation results in \secref{sec:sgps_sim} will show that performances of the low complexity SGPS algorithm can closely approach those of ES.

\section{Numerical Results}\label{sec:simulation}

\begin{figure*}
\begin{align}\label{eq:lap_pas}
  \pklaptheta=\oneon{\sqrt{2}\sigma_{k}\left(1-\expx{-\sqrt{2}\pi/\sigma_{k}}\right)}
  \cdot\expx{\frac{-\sqrt{2}\abs{\theta-\theta_{k}}}{\sigma_{k}}},\qquad \fortext \ \theta\in \left[\theta_{k}-\pi,\theta_{k}+\pi\right]
\end{align}
\hrulefill
\end{figure*}

In this section, we present numerical simulations to evaluate performances of the proposed PR based massive MIMO transmission. We assume that the BS is equipped with the $128$-antenna ULA, and the antennas are spaced with a half wavelength distance. We set the AoA interval as $\cA=[-\pi/2,\pi/2]$. We consider the typical outdoor wireless propagation environments where the channel PAS can be modeled as the truncated Laplacian distribution \cite{Pedersen00stochastic,Cho10MIMO} given by \eqref{eq:lap_pas}, shown at the top of the next page, where $\sigma_{k}$ and $\theta_{k}$ represent the AS and the mean AoA of the \ith{k} UT's channel, respectively. We assume that channel ASs are the same for all the UTs so that $\sigma_{k}=\sigma$ $(\forall k)$. We assume that all the UTs are of equal distance from the BS, and set the large scale fading coefficients as $\beta_{k}=1$ $(\forall k)$. We assume that the UTs uniformly locate in a $120^{\circ}$ sector, i.e., the mean channel AoA $\theta_{k}$ is uniformly distributed in the angle interval $[-\pi/3,\pi/3]$ in radian. The channel covariance matrices of the UTs are generated according to the model given by \remref{rem:ula_dft_eig}, and we impose the constraint in \eqref{eq:trrk} for channel power normalization. We assume that the channel training SNR and the data transmission SNR are equal such that $\rhoup=\rhou=\rhod=\rho$.

\subsection{Performance of Robust Transmission}\label{sec:simulation-robust}

In this subsection, we employ the average MSE-SD metric to evaluate performances of the robust receiver and precoder developed in \secref{sec:robust tra}. Due to the UL-DL MMSE duality given in \coref{co:ul dl duality}, we only consider the UL transmission case for brevity.

We compare performances of the robust MMSE receiver given in \eqref{eq:rob_opt_rec} with those of the conventional receiver given in \eqref{eq:conv_mmse_re}. We consider the case with $K=10$, $\sigma=10^{\circ}$, and the mean channel AoAs of the UTs from UT 1 to UT 10 are \begin{align*}
  &[0.6592,0.8499,-0.7812,0.8658,0.2772,\ntb
  &\qquad-0.8429,-0.4639,0.0982,0.9582,0.9737]
\end{align*}
in radian. We assume that the pilot length equals $\tau=5$, and consider two PR patterns. Specifically, the pilot indices that the UTs use from UT 1 to UT 10 are $[1,1,2,2,3,3,4,4,5,5]$ and $[1,2,3,4,3,5,4,5,5,3]$ for PR patterns A and B, respectively.\footnote{PR pattern B is determined by the SGPS algorithm proposed in \secref{subsec:sgps}, and we arbitrarily set a PR pattern as pattern A for comparison. We employ such setting to exemplify that pilot scheduling is crucial to the transmission performance.} In \figref{fig:mix}, we plot the average MSE-SD performances of the robust MMSE receiver (using the true and the estimated channel covariance matrices that are obtained via averaging over $100$ samples, respectively) and the conventional receiver. The lower bound of the average MSE-SD achieved by the robust MMSE receiver given in \lmref{lm:ammse lb} is also shown. We can have the following observations: 1) the average MSE-SD performance loss using the estimated channel covariance matrices compared with true channel covariance matrices can be almost neglected; 2) the robust MMSE receiver outperforms the conventional receiver, especially in high SNR regime where pilot interference dominates; 3) compared with the robust MMSE receiver, the conventional receiver is quite sensitive to the channel estimation error, and increasing the SNR may result in additional MSE-SD for the conventional receiver; 4) the closed-form lower bound of the average MSE-SD given in \lmref{lm:ammse lb} is tight over a wide SNR region for different PR patterns; 5) pilot scheduling is crucial to the data transmission performance.

\begin{figure}
\includegraphics[width=\figdoucolwid]{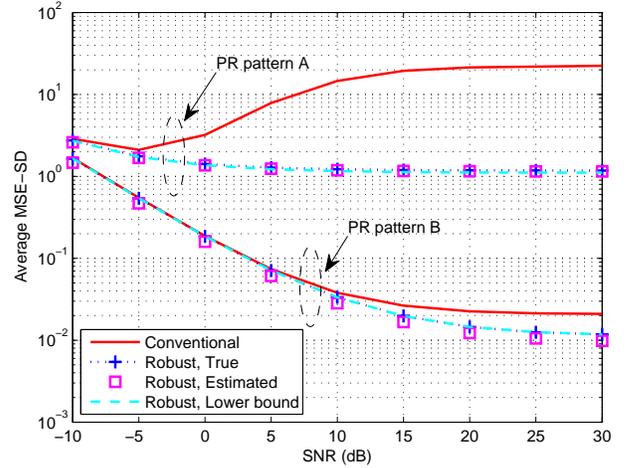}
\centering
\caption{Comparison of the average MSE-SD performances between the robust MMSE receiver (using the true and the estimated channel covariance matrices that are obtained via averaging over $100$ samples, respectively) and the conventional receiver. Results are shown versus the SNR for two specific PR patterns with $K=10$, $\tau=5$ and $\sigma=10^{\circ}$. The lower bound of the average MSE-SD achieved by the robust MMSE receiver is also depicted.}\label{fig:mix}
\end{figure}

\subsection{Performance of SGPS Algorithm}\label{sec:sgps_sim}

In this subsection, we evaluate performances of the SGPS algorithm, and compare them with those of ES. In \figref{fig:cha_mse_comp} and \figref{fig:sche_sig_mse}, we plot the MSE-CE metric in \eqref{eq:mse_est_obj} and the average MSE-SD metric in \eqref{eq:average_mmse_lb} versus the pilot length for different values of SNR with $K=10$ and $\sigma=10^{\circ}$, respectively. It can be observed that, in both cases the performances of the SGPS algorithm closely approach those of ES over a wide SNR region for different values of pilot length.

\begin{figure}
\includegraphics[width=\figdoucolwid]{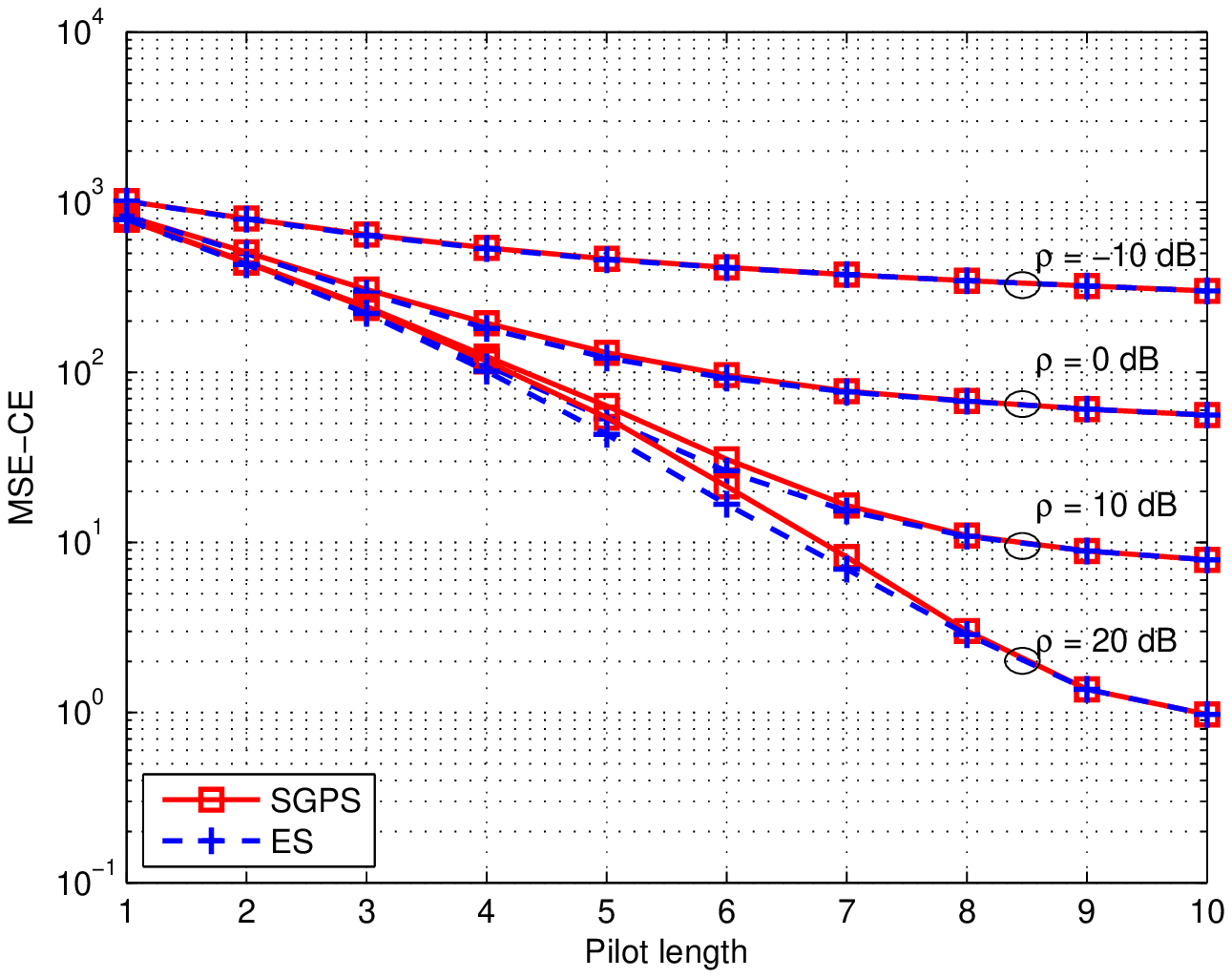}
\caption{Comparison of the MSE-CE performances between the SGPS algorithm and ES. Results are shown versus the pilot length for different values of SNR $\rho$ with $K=10$ and $\sigma=10^{\circ}$.}
\label{fig:cha_mse_comp}
\end{figure}

\begin{figure}
\includegraphics[width=\figdoucolwid]{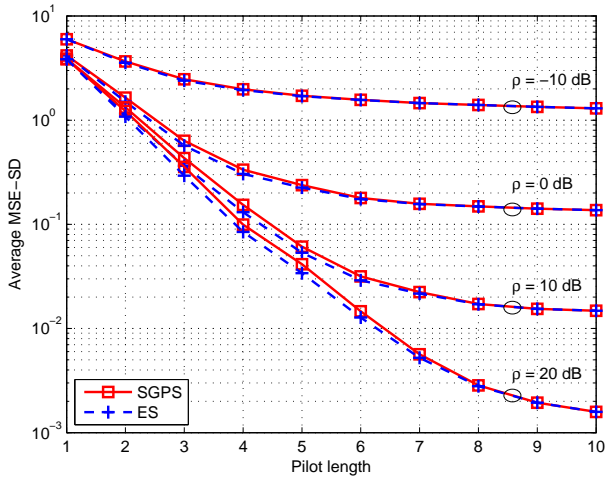}
\caption{Comparison of the average MSE-SD performances between the SGPS algorithm and ES. Results are shown versus the pilot length for different values of SNR $\rho$ with $K=10$ and $\sigma=10^{\circ}$.}
\label{fig:sche_sig_mse}
\end{figure}

\subsection{Net Spectral Efficiency Comparison}\label{subsec:net_rate}

In this subsection, we compare the \emph{net} spectral efficiency performance between the proposed PR scheme and the conventional OT scheme. The net spectral efficiency is given by
\begin{equation}\label{eq:net_spec_effi}
  \upnet{R}=\left(1-\frac{\tau}{T}\right)R^{\textrm{ach}}
\end{equation}
and the achievable rate $R^{\textrm{ach}}$ can be set as the UL sum achievable rate $\upusum{R}$, or the DL sum achievable rate $\updsum{R}$, or the weighted summation of $\upusum{R}$ and $\updsum{R}$. The UL sum achievable rate $\upusum{R}$ \cite{Hassibi03How,Hoydis13Massive} is given by \eqref{eq:ul sum rate}, shown at the top of the next page, where $\bw_{k}$ is the \ith{k} column of the UL receiver matrix $\bW$ given in \eqref{eq:rob_opt_rec}. The DL sum achievable rate $\updsum{R}$ \cite{Jose11Pilot,Hoydis13Massive} is given by \eqref{eq:dl sum rate}, shown at the top of the next page, where $\bb_{k}$ is the \ith{k} column of the DL precoding matrix $\bB$ given in \eqref{eq:rob_opt_pre}, and $\alpha$ is the power scaling performed at the UTs given in \eqref{eq:alpopt}.

\begin{figure*}
\begin{equation}\label{eq:ul sum rate}
  \upusum{R}=\sum_{k=1}^{K}{\expect{\logtwo{1+\frac{\abs{\bw_{k}^{T}\hatbg_{k}}^{2}}{\bw_{k}^{T}\left(\sum_{m\neq k}{\hatbg_{m}\hatbg_{m}^{H}}+\sum_{n=1}^{K}{\bR_{\tildebg_{n}}}+\invrho\bI\right)\bw_{k}^{*}}}}}
\end{equation}
\begin{equation}\label{eq:dl sum rate}
  \updsum{R}=\sum_{k=1}^{K}\logtwo{1+\frac{\abs{\expect{\alpha\bg_{k}^{T}\bb_{k}}}^{2}}
  {\sum_{m=1}^{K}\expect{\alpha^{2}\abs{\bg_{k}^{T}\bb_{m}}^{2}}-\abs{\expect{\alpha\bg_{k}^{T}\bb_{k}}}^{2}+\invrho\expect{\alpha^{2}}}}
\end{equation}
\hrulefill
\end{figure*}

For the PR scheme, we consider a dynamic pilot length strategy. Specifically, for a given UT set, the achievable rates in \eqref{eq:ul sum rate} and \eqref{eq:dl sum rate} can be obtained for arbitrary pilot length $\tau(<K)$ with pilot scheduling performed by the SGPS algorithm, and then the optimal pilot length and the net spectral efficiency can be obtained. While for the OT scheme, the pilot length $\tau$ is set as $\tau=K$ if $K\leq T/2$, or $\tau=\lfloor T/2\rfloor$ if $K>T/2$ where only $\lfloor T/2\rfloor$ UTs are serviced simultaneously \cite{Marzetta06How}.

The UL net spectral efficiency performances of the PR scheme and the OT scheme are compared in \figref{fig:ul_spec_effi_doppler} and \figref{fig:ul_spec_effi_snr}, while the DL net spectral efficiency performances are compared in \figref{fig:dl_spec_effi_doppler} and \figref{fig:dl_spec_effi_snr}, with $K=10$. It can be observed that, the proposed PR scheme shows performance gains over the conventional OT scheme in terms of the net spectral efficiency, and the gains become larger as the channel AS becomes smaller. Moreover, in the high SNR regime where the pilot interference dominates, and in the small coherence block length regime where the pilot overhead dominates, the proposed PR scheme provides significant performance gains. Specifically, for the case with $K=10$, $\sigma=2^{\circ}$, $\rho=20$ $\mathrm{dB}$ and $T=20$, the proposed PR scheme provides approximately $35$ $\mathrm{bits/s/Hz}$ net spectral efficiency gains over the conventional OT scheme for both the UL and DL data transmissions.

\begin{figure}
\includegraphics[width=\figdoucolwid]{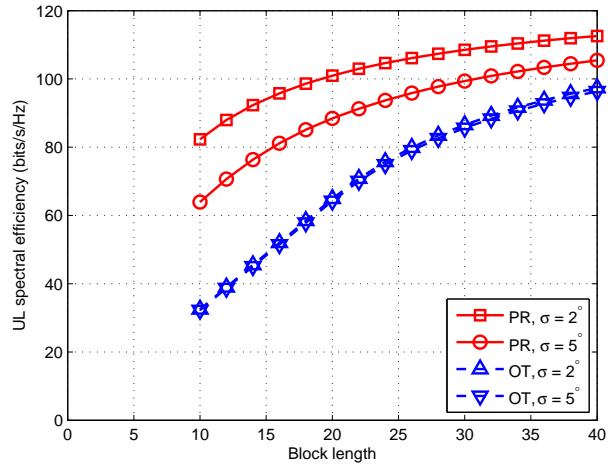}
\caption{Comparison of the UL spectral efficiency performances between the PR scheme and the OT scheme. Results are shown versus the coherence block length $T$ for different values of AS $\sigma$ with $K=10$ and $\rho=20$ $\mathrm{dB}$.}\label{fig:ul_spec_effi_doppler}
\end{figure}

\begin{figure}
\includegraphics[width=\figdoucolwid]{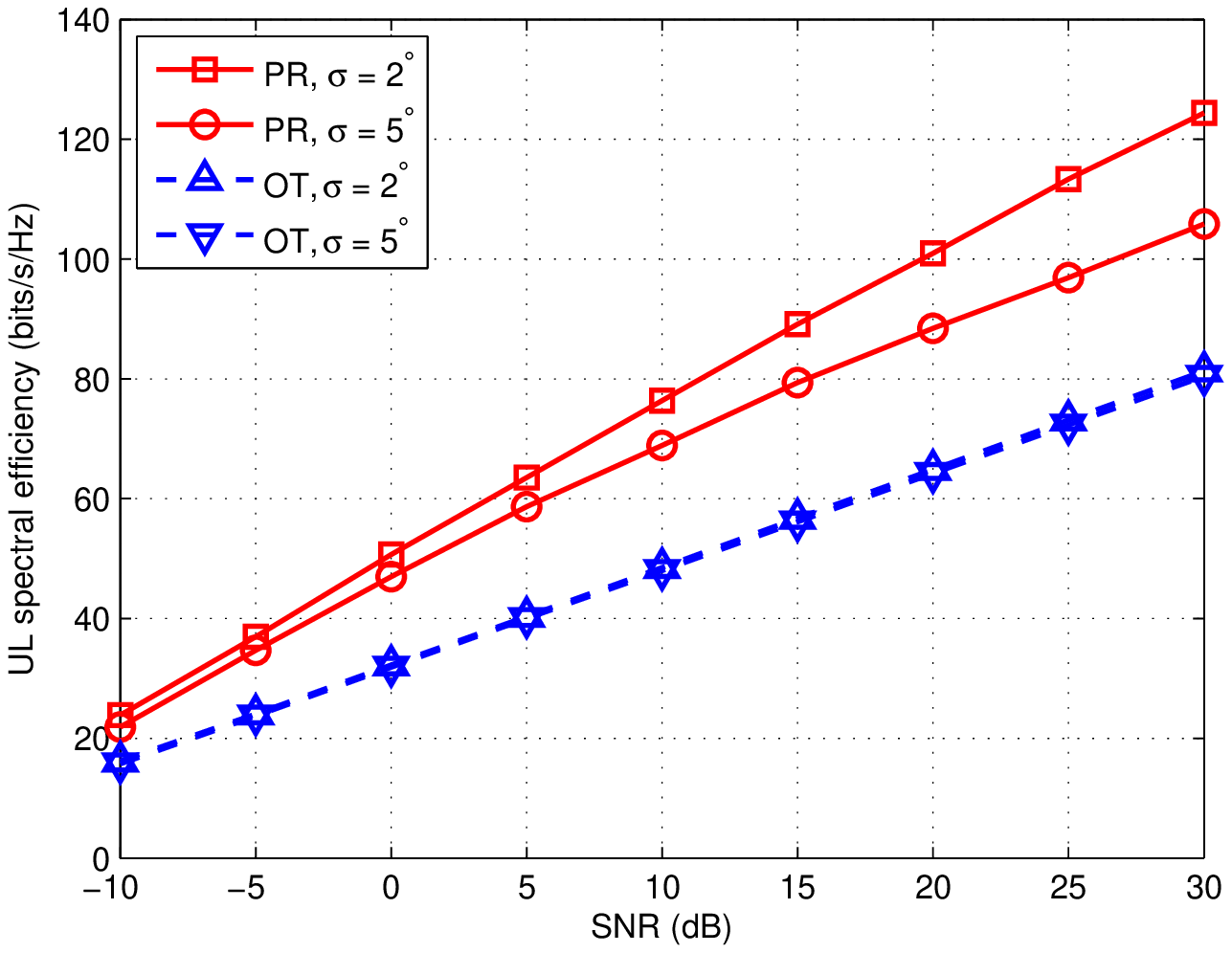}
\caption{Comparison of the UL spectral efficiency performances between the PR scheme and the OT scheme. Results are shown versus the SNR for different values of AS $\sigma$ with $K=10$ and $T=20$.}\label{fig:ul_spec_effi_snr}
\end{figure}

\begin{figure}
\includegraphics[width=\figdoucolwid]{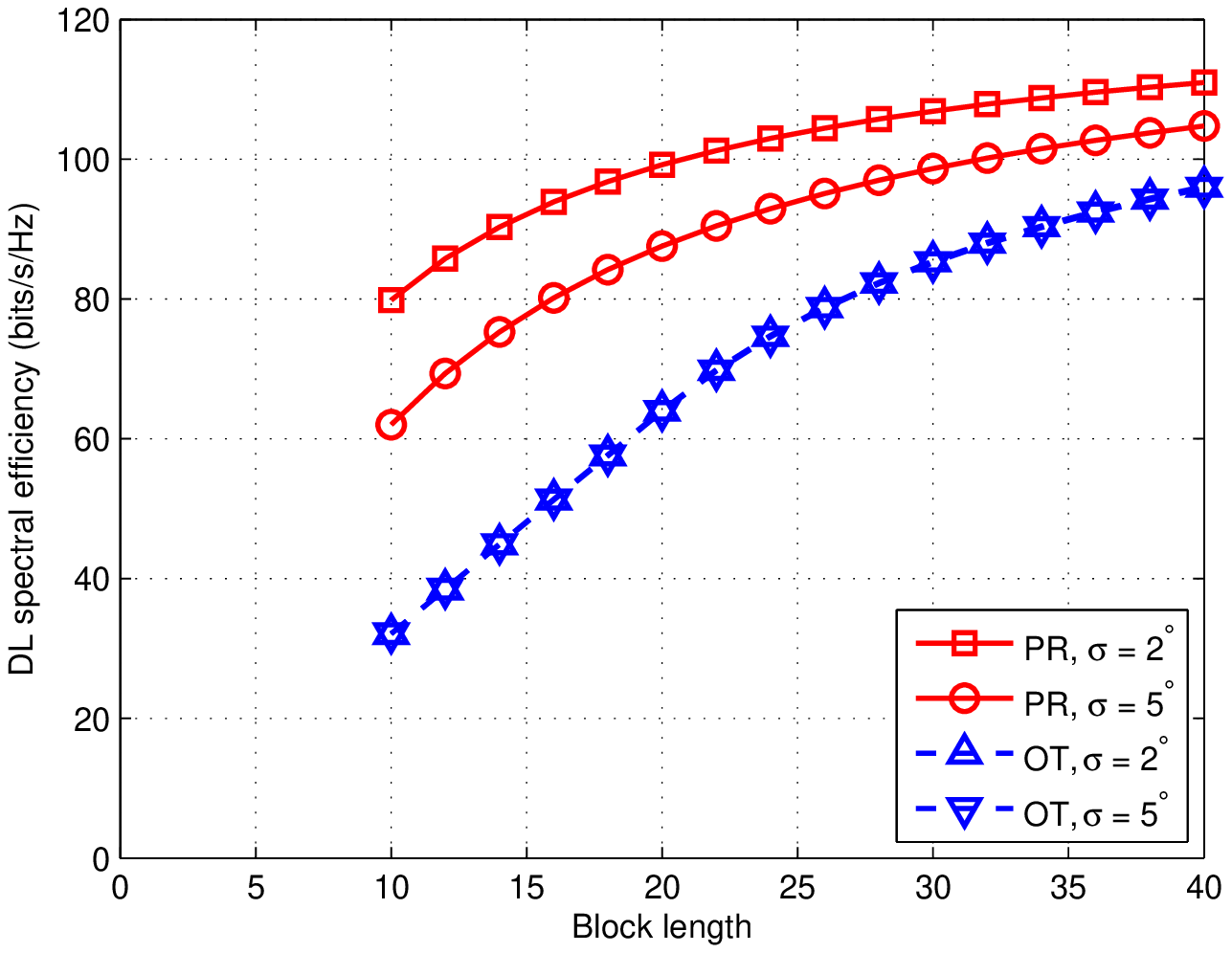}
\caption{Comparison of the DL spectral efficiency performances between the PR scheme and the OT scheme. Results are shown versus the coherence block length $T$ for different values of AS $\sigma$ with $K=10$ and $\rho=20$ $\mathrm{dB}$.}\label{fig:dl_spec_effi_doppler}
\end{figure}

\begin{figure}
\includegraphics[width=\figdoucolwid]{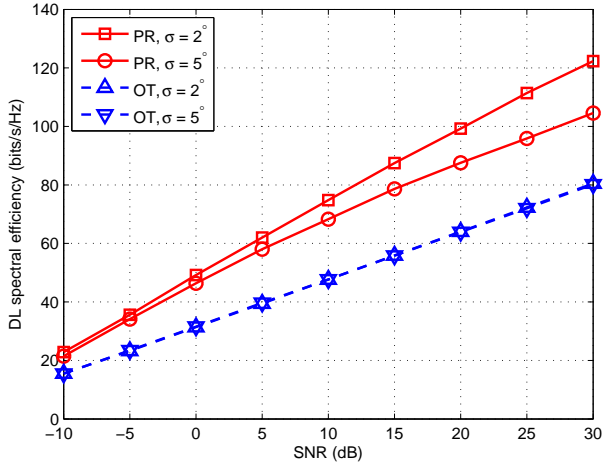}
\caption{Comparison of the DL spectral efficiency performances between the PR scheme and the OT scheme. Results are shown versus the SNR for different values of AS $\sigma$ with $K=10$ and $T=20$.}\label{fig:dl_spec_effi_snr}
\end{figure}

\section{Conclusion}\label{sec:conc}

In this paper, we proposed pilot reuse (PR) in single cell for massive MIMO transmission to reduce the pilot overhead. We exploited the fact that, in realistic outdoor wireless propagation environments where the BS is located at an elevated position, most of the channel power lies in a limited number of spatial directions compared with the whole massive MIMO channel dimension, and thus PR among UTs of spatial localization becomes feasible and beneficial. We first established the relationship between the channel covariance matrix and the channel PAS for the massive MIMO channels. Then we showed that, if channel AoA intervals of the UTs reusing the pilots are non-overlapping, then MSE of the channel estimation can be minimized. We also developed the robust multiuser UL receiver and DL precoder with the channel estimation error due to PR taken into account, and revealed the UL-DL MMSE duality between them. Moreover, we presented pilot scheduling under two MMSE related criteria, and proposed a low complexity pilot scheduling algorithm motivated by the channel AoA non-overlapping condition. The simulation results showed significant performance gains of the proposed PR scheme over the conventional OT scheme in terms of net spectral efficiency.

\appendices

\section{Proof of {\lmref{lm:Decomp_cov}}}\label{app:lemma_cov_decomp}
From the definition of $\bV$ in \eqref{eq:cov_eigv_mat}, we have
\begin{align}
  &\liminfty{M}\vecele{\bV^{H}\bV-\bI_{M}}{i,j}\ntb
  &\qquad=\liminfty{M}\oneon{M}\left(\vx{\vartheta\left(\psi_{i-1}\right)}\right)^{H}\vx{\vartheta\left(\psi_{j-1}\right)}-\delfunc{i-j}\notag\\
  &\qquad\equaa\delfunc{i-j}-\delfunc{i-j}=0
\end{align}
where (a) follows from \asref{ass:asymp_ortho_asv} and $\vartheta(\psi)$ is a strictly increasing function. This concludes the proof of \eqref{eq:cov_eig_uni}.

The proof of \eqref{eq:Decomp_cov} can be obtained as
\begin{align}
  &\liminfty{M}\vecele{\bR_{k}-\bV\diag{\br_{k}}\bV^{H}}{i,j}\ntb
  &\quad\equaa\liminfty{M}\vecele{\bR_{k}}{i,j}-\ntb
  &\qquad\liminfty{M}\vecele{\oneon{M}\sum_{m=1}^{M}\vecele{\br_{k}}{m}
  \vx{\vartheta\left(\psi_{m-1}\right)}\bv^{H}\left(\vartheta\left(\psi_{m-1}\right)\right)
  }{i,j}\ntb
  &\quad\equab\liminfty{M}\vecele{\bR_{k}}{i,j}-\beta_{k}\liminfty{M}
  \sum_{m=1}^{M}\vecele{\vx{\vartheta\left(\psi_{m-1}\right)}}{i}\ntb
  &\qquad\cdot\vecele{\vx{\vartheta\left(\psi_{m-1}\right)}}{j}^{*}
  \mathsf{S}_{k}\left(\vartheta\left(\psi_{m-1}\right)\right)\left[\vartheta\left(\psi_{m}\right)-\vartheta\left(\psi_{m-1}\right)\right]\notag\\
  &\quad\equac\beta_{k}\int\limits_{\thetamin}^{\thetamax}\!
  {\vecele{\vx{\theta}}{i}\vecele{\vx{\theta}}{j}^{*}\mathsf{S}_{k}\left(\theta\right)\intdx{\theta}}\ntb
  &\qquad-\beta_{k}\int\limits_{\vartheta\left(\psi_{0}\right)}^{\vartheta\left(\psi_{M}\right)}\!
  {\vecele{\vx{\vartheta\left(\psi\right)}}{i}\vecele{\vx{\vartheta\left(\psi\right)}}{j}^{*}\mathsf{S}_{k}\left(\vartheta\left(\psi\right)\right)\intdx{\vartheta\left(\psi\right)}}\notag\\
  &\quad\equad\beta_{k}\int\limits_{\thetamin}^{\thetamax}\!
  {\vecele{\vx{\theta}}{i}\vecele{\vx{\theta}}{j}^{*}\mathsf{S}_{k}\left(\theta\right)\intdx{\theta}}\ntb
  &\qquad-\beta_{k}\int\limits_{\thetamin}^{\thetamax}\!
  {\vecele{\vx{\theta}}{i}\vecele{\vx{\theta}}{j}^{*}\mathsf{S}_{k}\left(\theta\right)\intdx{\theta}}=
  0
\end{align}
where (a) follows from \eqref{eq:cov_eigv_mat}, (b) follows from \eqref{eq:cov_eigen_ele}, (c) follows from \eqref{eq:def_chan_cov} and the integral definition, (d) follows from that $\vartheta\left(\psi_{0}\right)=\vartheta\left(0\right)=\thetamin$ and $\vartheta\left(\psi_{M}\right)=\vartheta\left(1\right)=\thetamax$.

\section{Proof of {\thref{th:est_min_con}}}\label{app:est_min_con}

We start by presenting a lemma that is required in the following proof.
\begin{lemma}\label{lm:matrix product}
For $\bA\succeq\bzero$ and $\bB\succeq\bzero$, $\thcos{\bA}{\bB}=\pi/2$ is equivalent to $\bA\bB=\bzero$.
\end{lemma}
\begin{proof}
Recalling \eqref{eq:thcosdef}, $\thcos{\bA}{\bB}=\pi/2$ is equivalent to $\tr{\bA\bB}=0$. Furthermore, $\tr{\bA\bB}=0$ is equivalent to $\bA\bB=\bzero$ for $\bA\succeq\bzero$ and $\bB\succeq\bzero$ \cite[Prop. 4.26]{Seber08Matrix}. This concludes the proof.
\end{proof}

Now we proceed with the proof of the theorem. Due to the positive semi-definiteness of the covariance matrix, we can obtain
\begin{align}\label{eq:rbgk_min}
  \tr{\bR_{\tildebg_{k}}}&=\tr{\bR_{k}-\bR_{k}\invcovpi{k}\bR_{k}}\ntb
  &=\tr{\bR_{k}-\bR_{k}\inv{\left(\sum_{\ell\in\KPset{k}}{\bR_{\ell}}+\invtaurhoup\bI\right)}\bR_{k}}\notag\\
  &\geq\tr{\bR_{k}-\bR_{k}\inv{\left(\bR_{k}+\invtaurhoup\bI\right)}\bR_{k}}.
\end{align}

From \lmref{lm:matrix product} which states that $\thcos{\bR_{i}}{\bR_{j}}=\pi/2$ is equivalent to $\bR_{i}\bR_{j}=\bzero$, we can obtain
\begin{align}
  \covpi{k}\bR_{k}&=\left(\sum_{\ell\in\KPset{k}}{\bR_{\ell}+\invtaurhoup\bI}\right)\bR_{k}\ntb
  &=\left(\bR_{k}+\invtaurhoup\bI\right)\bR_{k}=\bR_{k}\left(\bR_{k}+\invtaurhoup\bI\right)
\end{align}
which indicates that
\begin{align}
  \invcovpi{k}\bR_{k}=\bR_{k}\inv{\left(\bR_{k}+\invtaurhoup\bI\right)}
  =\inv{\left(\bR_{k}+\invtaurhoup\bI\right)}\bR_{k}.\label{eq:briinvrpiup}
\end{align}
Substituting \eqref{eq:briinvrpiup} into \eqref{eq:cov_error}, we can obtain
\begin{equation}\label{eq:rbgk_sim}
  \tr{\bR_{\tildebg_{k}}}=\tr{\bR_{k}-\bR_{k}\inv{\left(\bR_{k}+\invtaurhoup\bI\right)}\bR_{k}}
\end{equation}
which achieves the minimum in \eqref{eq:rbgk_min}. This concludes the proof.

\section{Proof of {\thref{th:ul bf}}}\label{app:ul bf}

The MSE-SD defined in \eqref{eq:mseu} can be simplified as
\begin{align}\label{eq:upu_eps_int1}
\upu{\epsilon}&=
\mathrm{tr}\Bigg\{\bW^{T}\left(\hatbG\hatbG^{H}+\sum_{k=1}^{K}{\bR_{\tildebg_{k}}}+\invrhou\bI\right)\bW^{*}\ntb
&\qquad\qquad+\bI-\bW^{T}\hatbG-\hatbG^{H}\bW^{*}\Bigg\}.
\end{align}
Note that $\upu{\epsilon}$ is convex with respect to $\bW$.

By setting the derivative of $\upu{\epsilon}$ with respect to $\bW^{*}$ \cite{Hjorungnes11Complex} to zero,
\begin{align}
  \ppd{\bW^{*}}\upu{\epsilon}&=\left(\hatbG\hatbG^{H}+\sum_{k=1}^{K}{\bR_{\tildebg_{k}}}+\invrhou\bI\right)^{T}\bW-\hatbG^{*}\notag\\
  &\equaa\left(\hatbG\hatbG^{H}+\sum_{k=1}^{K}{\bR_{\tildebg_{k}}}+\invrhou\bI\right)^{*}\bW-\hatbG^{*}=\bzero
\end{align}
where (a) follows from $\bA^{T}=\bA^{*}$ for the Hermitian matrix $\bA$, we can obtain
  \begin{equation}\label{eq:ul_rob_int1}
    \upopt{\bW}=\left[\inv{\left(\hatbG\hatbG^{H}+\sum_{k=1}^{K}{\bR_{\tildebg_{k}}}+\invrhou\bI\right)}\hatbG\right]^{*}.
  \end{equation}

Substituting \eqref{eq:ul_rob_int1} into \eqref{eq:upu_eps_int1}, we can obtain the corresponding MSE-SD as
\begin{align}
  \upumin{\epsilon}&=\tr{\bI-\hatbG^{H}\inv{\left(\hatbG\hatbG^{H}+\sum_{k=1}^{K}{\bR_{\tildebg_{k}}}+\invrhou\bI\right)}\hatbG}\notag\\
  &\equaa\tr{\inv{\left(\bI+\hatbG^{H}\inv{\left(\sum_{k=1}^{K}{\bR_{\tildebg_{k}}}+\invrhou\bI\right)}\hatbG\right)}}
\end{align}
where (a) follows from the Woodbury matrix inversion identity \cite[Prop. 15.3]{Seber08Matrix}. This concludes the proof.
\section{Proof of {\thref{th:dl bf}}}\label{app:dl bf}

We start by simplifying the MSE-SD defined in \eqref{eq:msed} as
\begin{align}\label{eq:msed_simp}
  \upd{\epsilon}&=\expect{\sqrnormt{\alpha\left(\bG^{T}\bB\upd{\ba}+\invsqrtrhod\bnd\right)-\upd{\ba}}}\ntb
  &=\expect{\sqrnormt{\left(\alpha\bG^{T}\bB-\bI\right)\upd{\ba}}}+\frac{\alpha^{2}K}{\rhod}\ntb
  &=\mathrm{tr}\Bigg\{\alpha^{2}\bB^{H}\left(\hatbG^{*}\hatbG^{T}+\sum_{k=1}^{K}\bR_{\tildebg_{k}}^{*}\right)\bB\ntb
  &\qquad\qquad-\alpha\hatbG^{T}\bB-\alpha\bB^{H}\hatbG^{*}\Bigg\}
  +\left(\frac{\alpha^{2}}{\rhod}+1\right)K.
\end{align}
The simplified objective function in \eqref{eq:msed_simp} is non-convex with respect to $(\bB,\alpha)$. We first show that there exists a global minimum for the problem \eqref{eq:dl opt prob} in the following lemma.

\begin{lemma}\label{lm:exist_glob_dl_pre}
For the problem \eqref{eq:dl opt prob}, there exists a global optimal solution.
\end{lemma}
\begin{proof}
The problem in \eqref{eq:dl opt prob} is equivalent to
\begin{align}\label{eq:dl opt prob_alt}
\begin{split}
    \minis{\bB}{\quad\minis{\alpha(\bB)}}\qquad&\upd{\epsilon}\left(\bB,\alpha\right)\\
    \st\qquad&\tr{\bB\bB^{H}}\leq K
\end{split}
\end{align}
and the optimal $\alpha$ for the inner unconstrained optimization problem can be readily obtained as
\begin{equation}
  \alpha=\frac{\tr{\hatbG^{T}\bB+\bB^{H}\hatbG^{*}}}{2\left[K/\rhod+\tr{\bB^{H}\left(\hatbG^{*}\hatbG^{T}+\sum_{k=1}^{K}\bR_{\tildebg_{k}}^{*}\right)\bB}\right]}.
\end{equation}
Then the problem \eqref{eq:dl opt prob_alt} is equivalent to
\begin{align}\label{eq:equiv_dl_opt_prob}
    &\mini{\bB}\qquad\upd{\epsilon}\left(\bB\right)\ntb
    &=K-\frac{\left[\tr{\hatbG^{T}\bB+\bB^{H}\hatbG^{*}}\right]^{2}}{4\left[K/\rhod+\tr{\bB^{H}\left(\hatbG^{*}\hatbG^{T}+\sum_{k=1}^{K}\bR_{\tildebg_{k}}^{*}\right)\bB}\right]}\ntb
    &\st\qquad\tr{\bB\bB^{H}}\leq K.
\end{align}
The feasible set of \eqref{eq:equiv_dl_opt_prob} given by $\{\bB:\tr{\bB\bB^{H}}\leq K\}$ is compact (closed and bounded), and the objective function of \eqref{eq:equiv_dl_opt_prob} is continuous over the feasible set. Thus, according to Weierstrass extreme value theorem \cite[Appx. E]{Horn12Matrix}, there exists a global minimum for the problem \eqref{eq:equiv_dl_opt_prob}, and so does the equivalent problem \eqref{eq:dl opt prob}.
\end{proof}

\lmref{lm:exist_glob_dl_pre} shows that there exists a global optimum for the problem \eqref{eq:dl opt prob}. Note that the global optimal solution should satisfy the Karush-Kuhn-Tucker (KKT) necessary conditions \cite{Boyd04Convex}. In the following we will seek out all the solutions that satisfy the KKT conditions and identify the optimal solution among them.

The Lagrangian associated with the problem \eqref{eq:dl opt prob} is
\begin{align}\label{eq:lag_msed}
  \mathcal{L}\left(\bB,\alpha,\lambda\right)=\upd{\epsilon}+\lambda\left(\tr{\bB\bB^{H}}-K\right)
\end{align}
where $\upd{\epsilon}$ is given in \eqref{eq:msed_simp}, and $\lambda$ is the Lagrange multiplier associated with the inequality constraint.

The KKT necessary conditions for the problem \eqref{eq:dl opt prob} can be obtained as \cite{Hjorungnes11Complex}
\begin{align}
  &\ppd{\bB^{*}}\mathcal{L}\left(\bB,\alpha,\lambda\right)
  =\alpha^{2}\left(\hatbG^{*}\hatbG^{T}+\sum_{k=1}^{K}\bR_{\tildebg_{k}}^{*}\right)\bB\ntb
  &\qquad\qquad-\alpha\hatbG^{*}+\lambda\bB=\bzero\label{eq:kkt pbpl}\\
  &\ppd{\alpha}\mathcal{L}\left(\bB,\alpha,\lambda\right)
  =2\alpha\tr{\bB^{H}\left(\hatbG^{*}\hatbG^{T}+\sum_{k=1}^{K}\bR_{\tildebg_{k}}^{*}\right)\bB}\ntb
  &\qquad\qquad-\tr{\hatbG^{T}\bB+\bB^{H}\hatbG^{*}}+\frac{2\alpha K}{\rhod}=0\label{eq:kkt palphapl}\\
  &\lambda\geq0,\quad\tr{\bB\bB^{H}}\leq K\\
  &\lambda\left(\tr{\bB\bB^{H}}-K\right)=0.\label{eq:kkt lambdakkh}
\end{align}

An obvious solution that satisfies the above KKT conditions is $(\alpha=0,\bB=\bzero,\lambda=0)$, and the corresponding $\upd{\epsilon}$ equals $K$. For the case with $\alpha\neq0$, \eqref{eq:kkt pbpl} is equivalent to
\begin{align}\label{eq:alphaggh}
  \hatbG^{*}=\alpha\left(\hatbG^{*}\hatbG^{T}+\sum_{k=1}^{K}\bR_{\tildebg_{k}}^{*}+\frac{\lambda}{\alpha^{2}}\bI\right)\bB
\end{align}
which leads to
\begin{align}\label{eq:hatbghb}
\hatbG^{T}\bB&=\bB^{H}\hatbG^{*}\ntb
&\qquad=\alpha\bB^{H}\left(\hatbG^{*}\hatbG^{T}+\sum_{k=1}^{K}\bR_{\tildebg_{k}}^{*}+\frac{\lambda}{\alpha^{2}}\bI\right)\bB.
\end{align}

Combining \eqref{eq:hatbghb} with \eqref{eq:kkt palphapl} yields
\begin{align}\label{eq:condtion krhod}
  \frac{\alpha^{2}K}{\rhod}=\lambda\tr{\bB\bB^{H}}.
\end{align}
Substituting \eqref{eq:condtion krhod} into \eqref{eq:kkt lambdakkh}, we can obtain
\begin{align}\label{eq:lambdaalpha}
  \lambda=\frac{\alpha^{2}}{\rhod}>0,\qquad\tr{\bB\bB^{H}}=K.
\end{align}
Substituting \eqref{eq:lambdaalpha} into \eqref{eq:alphaggh} yields
\begin{align}\label{eq:bbalpha}
  \bB&=\oneon{\alpha}\inv{\left(\hatbG^{*}\hatbG^{T}+\sum_{k=1}^{K}\bR_{\tildebg_{k}}^{*}+\invrhod\bI\right)}\hatbG^{*}\ntb
  &=\oneon{\alpha}\left[\inv{\left(\hatbG\hatbG^{H}+\sum_{k=1}^{K}\bR_{\tildebg_{k}}+\invrhod\bI\right)}\hatbG\right]^{*}
\end{align}
where $\alpha$ is chosen to satisfy the constraint $\tr{\bB\bB^{H}}=K$.

Substituting \eqref{eq:bbalpha} into \eqref{eq:msed_simp}, we can obtain the corresponding MSE-SD as
\begin{subequations}
\begin{align}
  \upd{\epsilon}&=\tr{\bI-\hatbG^{T}\inv{\left(\hatbG^{*}\hatbG^{T}+\sum_{k=1}^{K}{\bR_{\tildebg_{k}}^{*}}+\invrhod\bI\right)}\hatbG^{*}}\notag\\
  &=\tr{\bI-\hatbG^{H}\inv{\left(\hatbG\hatbG^{H}+\sum_{k=1}^{K}{\bR_{\tildebg_{k}}}+\invrhod\bI\right)}\hatbG}\label{eq:updsum_med}\\
  &=\tr{\inv{\left(\bI+\hatbG^{H}\inv{\left(\sum_{k=1}^{K}{\bR_{\tildebg_{k}}}+\invrhod\bI\right)}\hatbG\right)}}\label{eq:updsum_med_mod}
\end{align}
\end{subequations}
where \eqref{eq:updsum_med} follows from the trace identity $\tr{\bA}=\tr{\bA^{T}}$ \cite[Eq. (2.95)]{Hjorungnes11Complex} and $\bR_{\tildebg_{k}}$ is Hermitian, \eqref{eq:updsum_med_mod} follows from the Woodbury matrix inversion identity \cite[Prop. 15.3]{Seber08Matrix}. Note that the MSE-SD in \eqref{eq:updsum_med} is smaller than $K$ that previously obtained from the solution $(\alpha=0,\bB=\bzero,\lambda=0)$. Therefore, we obtain that the precoder given by \eqref{eq:bbalpha} is optimal. This concludes the proof.

\section{Proof of {\lmref{lm:ammse lb}}}\label{app:ammse lb}

Via invoking the matrix-valued Jensen's inequality which states that $\expect{\inv{\bA}}\succeq\inv{\left(\expect{\bA}\right)}$ for $\bA\succ\bzero$ \cite[Prop. 21.64]{Seber08Matrix}, we can obtain
\begin{align}
  \expect{\uptmin{\epsilon}}
  &\geq\tr{\inv{\left(\bI_{K}+\expect{\hatbG^{H}\inv{\left(\uptneff{\bR}\right)}\hatbG}\right)}}\ntb
  &=\tr{\inv{\left(\bI_{K}+\bOmega\right)}}
\end{align}
and $\bOmega$ satisfies that
\begin{align}
  \vecele{\bOmega}{i,j}&=\vecele{\expect{\hatbG^{H}\inv{\left(\uptneff{\bR}\right)}\hatbG}}{i,j}\ntb
  &=\expect{\hatbg_{i}^{H}\inv{\left(\uptneff{\bR}\right)}\hatbg_{j}}\notag\\
  &\equaa\expect{\left(\ypiup\right)^{H}\invcovpi{i}\bR_{i}\inv{\left(\uptneff{\bR}\right)}\bR_{j}\invcovpi{j}\ypjup}\ntb
  &\equab\tr{\invcovpi{i}\bR_{i}\inv{\left(\uptneff{\bR}\right)}\bR_{j}}\delfunc{\Pindex{i}-\Pindex{j}}
\end{align}
where (a) follows from \eqref{eq:est_cha}, and (b) follows from \eqref{eq:ypkup}. This concludes the proof.

\section{Proof of {\thref{th:pilot reuse}}}\label{app:pilot reuse}

Via invoking the Schwartz inequality as in \cite[Lemma 1]{Ohno04Capacity}, we can obtain
\begin{equation}
  \upnot{\epsilon}{t,alb}=\tr{\inv{\left(\bI_{K}+\bOmega\right)}}\geq\sum_{i=1}^{K}\oneon{1+\vecele{\bOmega}{i,i}}
\end{equation}
where the equality is attained if and only if $\bOmega$ is diagonal.

Recalling \lmref{lm:matrix product} which states that $\bR_{i}\bR_{j}=\bzero$ is equivalent to $\thcos{\bR_{i}}{\bR_{j}}=\pi/2$, we only have to show that if $\bR_{i}\bR_{j}=\bzero$ for $\forall i\neq j$ and $\Pindex{i}=\Pindex{j}$, then $\bOmega$ is diagonal, i.e., $\vecele{\bOmega}{i,j}=0$ for $\forall i\neq j$ and $\Pindex{i}=\Pindex{j}$.

For $\forall i\neq j$ and $\Pindex{i}=\Pindex{j}$, if $\bR_{i}\bR_{j}=\bzero$, then
\begin{align}\label{eq:oij}
  \vecele{\bOmega}{i,j}&=\tr{\invcovpi{i}\bR_{i}\inv{\left(\uptneff{\bR}\right)}\bR_{j}}\ntb
  &\equaa\tr{\bR_{i}\inv{\left(\bR_{i}+\invtaurhoup\bI\right)}\inv{\left(\uptneff{\bR}\right)}\bR_{j}}\ntb
  &=\tr{\bR_{j}\bR_{i}\inv{\left(\bR_{i}+\invtaurhoup\bI\right)}\inv{\left(\uptneff{\bR}\right)}}
  =0
\end{align}
where (a) follows from \eqref{eq:briinvrpiup}.

Furthermore, if $\bR_{i}\bR_{j}=\bzero$ for $\forall i\neq j$ and $\Pindex{i}=\Pindex{j}$, then diagonal elements of $\bOmega$ reduces to
\begin{align}\label{eq:omegaii}
  &\vecele{\bOmega}{i,i}
  =\mathrm{tr}\Bigg\{\inv{\left(\bR_{i}+\invtaurhoup\bI\right)}\bR_{i}\ntb
  &\quad\cdot\inv{\left[\sum_{k=1}^{K}\left(\bR_{k}-\bR_{k}\inv{\left(\bR_{k}+\invtaurhoup\bI\right)}\bR_{k}\right)+\invrhot\bI\right]}
  \bR_{i}\Bigg\}
\end{align}
via invoking \eqref{eq:briinvrpiup}. This concludes the proof.

\end{document}